\documentclass[11pt,a4paper]{article}
\usepackage{amsmath}\usepackage{epsf,amsfonts,amsthm}\usepackage{amscd,amssymb}
\usepackage{xcolor,epic,eepic}\usepackage{epsfig}
\usepackage{fontenc,indentfirst, delarray,amsfonts,amsmath,amssymb}
\usepackage{rotating}
\usepackage{mathdots}
\usepackage[T1]{fontenc}
\usepackage[matrix,arrow,curve]{xy}
\usepackage{amsmath}
\usepackage{amssymb}
\usepackage{amsthm}
\usepackage{amscd}
\usepackage{amsfonts}
\usepackage{graphicx}
\usepackage{fancyhdr}
\usepackage{dsfont,texdraw}
\usepackage{amsmath}\usepackage{epsf,amsfonts,amsthm}\usepackage{amscd,amssymb}
\usepackage{xcolor,epic,eepic}\usepackage{epsfig}
\usepackage{fontenc,indentfirst, delarray,amsfonts,amsmath,amssymb}
\usepackage{rotating}
\usepackage{amscd}
\usepackage{amsmath}
\usepackage{amsthm}
\usepackage{mathrsfs}
\usepackage{latexsym}
\usepackage{amssymb}
\usepackage{amsfonts}
\theoremstyle{plain}

\usepackage{tikz}
\usepackage{tikz-cd}
\usetikzlibrary{arrows,chains,matrix,positioning,scopes,snakes}
\makeatletter
\tikzset{join/.code=\tikzset{after node path={%
\ifx\tikzchainprevious\pgfutil@empty\else(\tikzchainprevious)%
edge[every join]#1(\tikzchaincurrent)\fi}}}
\makeatother

\tikzset{>=stealth',every on chain/.append style={join},
         every join/.style={->}}
\tikzset{
    >=stealth',
    punkt/.style={
           rectangle,
           rounded corners,
           draw=black, very thick,
           text width=6.5em,
           minimum height=2em,
           text centered},
    pil/.style={
           ->,
           thick,
           shorten <=2pt,
           shorten >=2pt,}
}

\usepackage[T1]{fontenc}
\newcommand{\bee}{\begin{enumerate}}
\newcommand{\eee}{\end{enumerate}}
\newcommand{\benn}{\begin{equation*}}
\newcommand{\eenn}{\end{equation*}}
\newcommand{\be}{\begin{equation}}
\newcommand{\ee}{\end{equation}}
\newcommand{\bean}{\begin{eqnarray}}
\newcommand{\eean}{\end{eqnarray}}
\newcommand{\bea}{\begin{eqnarray*}}
\newcommand{\eea}{\end{eqnarray*}}

\newcommand{\Ci}{C^{\infty}}

\newcommand{\N}{\mathbb{N}}
\newcommand{\Z}{\mathbb{Z}}
\newcommand{\R}{\mathbb{R}}

\newcommand{\C}{\mathbb{C}}

\newcommand{\op}[1]{\!\!\mathop{\rm ~#1}\nolimits}

\newcommand{\cN}{{\cal N}}
\newcommand{\Zn}{\Z_2^n}
\newcommand{\cF}{{\cal F}}

\mathchardef\za="710B  
\mathchardef\zb="710C  
\mathchardef\zg="710D  
\mathchardef\zd="710E  
\mathchardef\zve="710F 
\mathchardef\zz="7110  
\mathchardef\zh="7111  
\mathchardef\zy="7112 

\mathchardef\zi="7113  
\mathchardef\zk="7114  
\mathchardef\zl="7115  
\mathchardef\zm="7116  
\mathchardef\zn="7117  
\mathchardef\zx="7118  
\mathchardef\zp="7119  
\mathchardef\zr="711A  
\mathchardef\zs="711B  
\mathchardef\zt="711C  
\mathchardef\zu="711D  
\mathchardef\zf="711E 
\mathchardef\zq="711F  
\mathchardef\zc="7120  
\mathchardef\zw="7121  
\mathchardef\ze="7122  
\mathchardef\zvy="7123  
\mathchardef\zvw="7124  
\mathchardef\zvr="7125 
\mathchardef\zvs="7126 
\mathchardef\zvf="7127  
\mathchardef\zG="7000  
\mathchardef\zD="7001  
\mathchardef\zY="7002  
\mathchardef\zL="7003  
\mathchardef\zX="7004  
\mathchardef\zP="7005  
\mathchardef\zS="7006  
\mathchardef\zU="7007  
\mathchardef\zF="7008  
\mathchardef\zW="700A  

\newcommand{\cyclic}{\mathop{\kern0.9ex{{+}
\kern-2.15ex\raise-.25ex\hbox{\Large\hbox{$\circlearrowright$}}}}\limits}

\newcommand{\cE}{{\cal E}}

 \newcommand{\cV}{{\cal V}}
 \newcommand{\cH}{{\cal H}}

 \newcommand{\cA}{{\cal A}}
 \newcommand{\cM}{{\cal M}}
 
 \newcommand{\cO}{{\cal O}}
 \newcommand{\cB}{{\cal B}}
 \newcommand{\cT}{{\cal T}}
 
 \newcommand{\cI}{{\cal I}}

 \newcommand{\cG}{{\cal G}}

\newtheorem{rem}{Remark}
\newtheorem{theo}[rem]{Theorem}
\newtheorem{prop}[rem]{Proposition}
\newtheorem{lem}[rem]{Lemma}
\newtheorem{cor}[rem]{Corollary}

\newtheorem{defi}[rem]{Definition}

\newtheorem{theorem}[rem]{Theorem}

\newcommand{\h}{\op{Hom}}
\newcommand{\0}{\otimes}

\newcommand{\id}{\op{id}}

\DeclareMathAlphabet{\mathpzc}{OT1}{pzc}{m}{it}

\newcommand{\Hom}{\mathrm{Hom}}

\begin{document}
\title{\bf Products in the category of $\Z_2 ^n$-manifolds}
\date{}
\author{Andrew Bruce and Norbert Poncin\footnote{University of Luxembourg, Mathematics Research Unit, 4364 Esch-sur-Alzette, Luxembourg, andrew.bruce@uni.lu, norbert.poncin@uni.lu}}
\maketitle

\begin{abstract} We prove that the category of $\Zn$-manifolds has all finite products. Further, we show that a $\Zn$-manifold (resp., a $\Zn$-morphism) can be reconstructed from its algebra of global $\Zn$-functions (resp., from its algebra morphism between global $\Zn$-function algebras). These results are of importance in the study of $\Zn$ Lie groups. The investigation is all the more challenging, since the completed tensor product of the structure sheafs of two $\Zn$-manifolds is not a sheaf. We rely on a number of results on (pre)sheaves of topological algebras, which we establish in the appendix.\end{abstract}

\tableofcontents

\section{Introduction}

$\Z_2^n$-Geometry is an emerging framework in mathematics and mathematical physics, which has been introduced in the foundational papers \cite{CGPa} and \cite{COP}. This non-trivial extension of standard Supergeometry allows for $\Z_2^n$-gradings, where $$\Z_2^n=\Z_2^{\times n}=\Z_2\times\ldots\times\Z_2\quad\text{and}\quad n\in\N\;.$$ The corresponding $\Z_2^n$-commutation rule for coordinates $(u^A)_A$ with degrees $\deg u^A\in\Z_2^n$ does not use the product of the (obviously defined) parities, but the scalar product $\langle-,-\rangle$ of $\Z_2^n\,$: \be\label{ZCR}u^Au^B=(-1)^{\langle\op{deg}u^A,\op{deg}u^B\rangle}u^Bu^A\;.\ee A brief description of the category $\Z_2^n\tt Man$ of $\Z_2^n$-manifolds $\cM=(M,\cO_M)$ and morphisms $\zF=(\zf,\zf^*)$ between them can be found in Section \ref{Summary}. For $n=1$, one recovers the category $\tt SMan$ of supermanifolds. A survey on $\Z_2^n$-Geometry is available in \cite{Pon}. The differential calculus and the splitting theorem for $\Z_2^n$-manifolds have been investigated in \cite{CKP} and \cite{CGPb}, respectively. In the introduction of \cite{BP1}, the reader finds motivations for the study of $\Z_2^n$-Geometry. The present paper uses the main results of \cite{BP1}.\medskip

Applications of $\Z_2^n$-Geometry are based in particular on $\Z_2^n$ Lie groups and their actions on $\Z_2^n$-manifolds (supergravity), on $\Z_2^n$ vector bundles and their sections ($\Z_2^n$ Lie algebroids), on the internal Hom functor in $\Z_2^n\tt Man$ ($\Z^n$-gradings and $\Z_2^n$-parities in field theory), ... All these notions rely themselves on {\it products} in the category $\Z_2^n\tt Man$. On the other hand, a comparison of different approaches to $\Z_2^n$ vector bundles is more challenging than in the supercase \cite{BCC}. A generalization to $\Z_2^n$-manifolds of the Schwarz-Voronov-Molotkov embedding is needed. This extension, which embeds $\Z_2^n\tt Man$ into the category of contravariant functors from $\Z_2^n$-points to a specific category of Fr\'echet manifolds, uses the {\it reconstructions} of $\Z_2^n$-manifolds and $\Z_2^n$-morphisms from the $\Z_2^n$-commutative associative unital $\R$-algebras of global $\Z_2^n$-functions and the $\Z_2^n$-graded unital algebra morphisms between them, respectively.\medskip

The existence of categorical products and the mentioned reconstruction theorems are the main topics of the present paper. The text is organized as follows. Section \ref{ReconstructionSection} contains the proofs of the above-mentioned $\Z_2^n$ reconstruction results. The definition of a product $\Z_2^n$-manifold and the proof of its meaningfulness are rather obvious, see Definition \ref{ProdZnMan}. However, the proof of the existence of categorical products in $\Z_2^n\tt Man$ is quite tough. It relies on the generalization of the well-known isomorphism of topological vector spaces $$\Ci(\zW')\widehat\0\,\Ci(\zW'')\simeq\Ci(\zW'\times \zW'')$$ (for open subsets $\zW'\subset \R^p, \zW''\subset \R^r$) to an isomorphism of locally convex topological algebras of formal power series \be\label{Intro1}\Ci(\zW')[[\xi]]\,\widehat\0\;\Ci(\zW'')[[\zh]]\simeq\Ci(\zW'\times\zW'')[[\xi,\zh]]\ee (for $\Z_2^n$-domains ${\cal U}^{p|\mathbf{q}}=(\zW',\Ci_{\zW'}[[\xi]])$ and ${\cal V}^{r|\mathbf{s}}=(\zW'',\Ci_{\zW''}[[\zh]])$, with $$\xi=(\xi_1,\ldots,\xi_{N(\mathbf{q})})\quad\text{and}\quad\zh=(\zh_1,\ldots,\zh_{N(\mathbf{s})}))\;,$$ see Theorem \ref{FundaIsom}. The issue here is the formal power series, which replace the polynomials of standard Supergeometry. Moreover, if $\cM=(M,\cO_M)$ and $\cN=(N,\cO_N)$ are two $\Z_2^n$-manifolds, and $\cM\times\cN=(M\times N,$ $\cO_{M\times N})$ is their product $\Z_2^n$-manifold, one gets from \eqref{Intro1} that, for an open subset $u\times v\subset M\times N$ of the basis $\frak B$ made of products of $\Z_2^n$-chart domains, we have $$\cE(u\times v)\simeq \cO_{M\times N}(u\times v)\;,$$ where $\cE$ is the $\frak B$-presheaf $$\cE(u\times v)=\cO_M(u)\,\widehat\0\;\cO_N(v)\;.$$ Let now $\bar\cF$ be the standard extension of the $\cB$-presheaf $\cF$ that assigns to any open subset $U\times V\subset M\times N$ (where $U\subset M$ and $V\subset N$ are not necessarily chart domains) the algebra $$\cF(U\times V)=\cO_M(U)\,\widehat\0\;\cO_N(V)\;.$$ The presheaf $\bar\cF$ and the sheaf $\cO_{M\times N}$ are thus two extensions of the $\frak B$-presheaf $\cE$. However, this does not mean that $\bar\cF\simeq\cO_{M\times N}$ and that $\bar\cF$ is a sheaf. Indeed, $\frak B$-sheaves have unique extensions, but $\frak B$-presheaves do not. Also the reconstruction results mentioned above do not allow us to prove that $\bar\cF$ is a sheaf. Hence, we prove the next best result, i.e., the existence of an isomorphism of sheaves of algebras \be\label{Intro2}\cO_{M\times N}\simeq \bar\cF^+\ee between the structure sheaf of the product $\Z_2^n$-manifold and the sheafification of the presheaf $\bar\cF$, see Theorem \ref{IsomExtShfifi}. In the case $n=1$, we thus recover the definition of a product supermanifold used in \cite{BBH}. The isomorphism \eqref{Intro2} allows us to prove the existence of all finite categorical products in $\Z_2^n\tt Man$, see Theorem \ref{FinProds}. The proof uses the results on sheafification and presheaves of locally convex topological algebras proven in Subsections \ref{SheafiPShTA} and \ref{SheafiDirIm} of the Appendix. Products of $\Z_2^n$-morphisms are obtained from the universality property of categorical products. They are explicitly described in Proposition \ref{PropProdZ2nMor}.

\section{$\Z_2^n$-manifolds and their morphisms}\label{Summary}

We denote by $\Z_2^n$ the cartesian product of $n$ copies of $\Z_2\,$. Further, we use systematically the fol\-low\-ing \emph{standard order} of the $2^n$ elements of $\Z_2^n$: first the even degrees are ordered lexicographically, then the odd ones are also ordered lexicographically. For example, $$ \Z_2^3=\{(0,0,0), (0,1,1), (1,0,1), (1,1,0), (0,0,1),(0,1,0), (1,0,0), (1,1,1)\}\;.$$

A $\Z_2^n$-domain has, aside from the usual coordinates $x=(x^1,\ldots,x^p)$ of degree $\deg x^i=0\in\Z_2^n$, also formal coordinates or parameters $\xi=(\xi^1,\ldots,\xi^Q)$ of non-zero degrees $\deg\xi^a\in\Z_2^n$. These coordinates $u=(x,\xi)$ commute according to the generalized sign rule \be\label{SignRule}u^Au^B=(-1)^{\langle\deg u^A,\deg u^B\rangle}u^Bu^A\;,\ee where $\langle-,-\rangle$ denotes the standard scalar product. For instance, $$\langle (0,1,1),(1,0,1)\rangle=1\;.$$ Observe that, in contrast with ordinary $\Z_2$- or super-domains, even coordinates may anticommute, odd coordinates may commute, and even nonzero degree coordinates are not nilpotent. Of course, for $n=1$, we recover the classical situation. We denote by $p$ the number of coordinates $x^i$ of degree 0, by $q_1$ the number of coordinates $\xi^a$ which have the first non-zero degree of $\Z_2^n$, and so on. We get that way a tuple $\mathbf{q}=(q_1,\ldots,q_N)\in\N^N$ with $N:=2^{n}-1$. The dimension of the considered $\Z_2^n$-domain is then given by $p|\mathbf{q}$. Clearly the $Q$ above is the sum $|\mathbf{q}|=\sum_{i=1}^Nq_i$.
\medskip

We recall the definition of a $\Z_2^n$-manifold.

\begin{defi}
A \emph{locally $\Z_2^n$-ringed space} is a pair $(M,\cO_M)$ made of a topological space $M$ and a sheaf of $\Z_2^n$-graded $\Z_2^n$-commutative $(\,$in the sense of \eqref{SignRule}$\,)$ associative unital $\R$-algebras over it, such that at every point $m\in M$ the stalk $\cO_{M,m}$ is a local graded ring.\smallskip

A smooth \emph{$\Z_2^n$-manifold} of dimension $p|\mathbf{q}$ is a \emph{locally $\Z_2^n$-ringed space} $\cM=(M,\cO_M)$, which is locally isomorphic to the smooth $\Z_2^n$-domain $\R^{p|\mathbf{q}}:=(\R^p,\Ci_{\R^p}[[\xi]])$, and whose underlying topological space $M$ is second-countable and Hausdorff. Sections of the structure sheaf $\,\C^{\infty}_{\R^p}[[\xi]]$ are \emph{formal power series} in the $\Z_2^n$-commutative parameters $\xi$, with coefficients in smooth functions:
$$
 \Ci_{\R^p}(U)[[\xi]]:=\left\{ \sum_{\alpha\in\N^{\times |\mathbf{q}|}} f_{\alpha}(x)\,\xi^{\alpha}\; |\; f_{\alpha}\in\Ci(U)\right\}\quad(U\,\text{open in}\;\,\R^p) \;.
$$

\emph{$\Z_2^n$-morphisms} between $\Z_2^n$-manifolds are just morphisms of $\Z_2^n$-ringed spaces, i.e., pairs
$\zF=(\phi,\phi^*):(M,\cO_M)\to (N,\cO_N)$ made of a continuous map $\phi:M\to N$ and a sheaf morphism $\phi^*:\cO_N\to\phi_*\cO_M$, i.e., a family of $\Z_2^n$-graded unital $\R$-algebra morphisms, which commute with restrictions and are defined, for any open $V\subset N$, by $$\phi^*_V:\cO_N(V)\to \cO_M(\phi^{-1}(V))\;.$$

We denote the category of $\Z_2^n$-manifolds and $\Z_2^n$-morphisms between them by $\Z_2^n{\tt Man}$.
\end{defi}

\begin{rem} Let us stress that the base space $M$ corresponds to the degree zero coordinates (and not to the even degree coordinates), and let us mention that it can be proven that the topological base space $M$ carries a natural smooth manifold structure of dimension $p$, that the continuous base map $\phi:M\to N$ is in fact smooth, and that the algebra morphisms $$\phi^*_m:\cO_{\phi(m)}\to\cO_{m}\quad(m\in M)\;$$ between stalks, which are induced by the considered $\Z_2^n$-morphism $\Phi:{\cal M}\to {\cal N}$, respect the unique homogeneous maximal ideals of the local graded rings $\cO_{\phi(m)}$ and $\cO_{m}$.
\end{rem}

\section{Reconstructions of $\Zn$-manifolds and $\Zn$-morphisms}\label{ReconstructionSection}

In this section, we reconstruct a $\Z_2^n$-manifold $(M,\cO_M)$ from the $\Z_2^n$-commutative unital algebra $\cO_M(M)$ of global sections of its function sheaf. We also reconstruct a $\Zn$-morphism $$\Phi=(\phi,\phi^*):(M,\cO_M)\to(N,\cO_N)$$ from its pullback $\Zn$-graded unital algebra morphism $$\phi^*_N:\cO_N(N)\to\cO_M(M)$$ between global sections.

\subsection{Reconstruction of the topological base space}

Algebraic characterizations of spaces can be traced back to I. Gel'fand and A. Kolmogoroff \cite{GK}. In that paper, compact topological spaces $K$ are characterized by the algebras $C^0(K)$ of continuous functions on them. In particular, the points $m$ of these spaces are identified with the maximal ideals $$I_m=\{f\in C^0(K): f(m)=0\}$$ of these algebras. A similar characterization holds for the points of second countable Hausdorff smooth manifolds. \medskip

Let $\cM=(M,\cO_M)$ be a $\Z_2^n$-manifold. We denote the maximal spectrum of $\cO(M)$ (subscript omitted) by $\op{Spm}(\cO(M))$ (we actually consider here the real maximal spectrum, in the sense that the quotient $\cO(M)/\zm$ by an ideal $\zm$ in the  spectrum is isomorphic to the field $\R$ of real numbers). Note that any $m\in M$ induces a map $$\ze_m:\cO(M)\ni f\mapsto (\ze_M f)(m)\in\R\;,$$ which is referred to as the {\it evaluation map} at $m$ and is a $\Z_2^n$-graded unital $\R$-algebra morphism $$\ze_m\in\h_{\Z_2^n{\tt UAlg}}(\cO(M),\R)\;.$$ The kernel $$\zm_m=\ker \ze_m=\{f\in\cO(M):(\ze_M f)(m)=0\}\in\op{Spm}(\cO(M))$$ is a maximal ideal. More generally, the kernel of an arbitrary algebra morphism $$\psi\in\h_{\Z_2^n{\tt UAlg}}(\cO(M),\R)$$ is a maximal ideal, since $\cO(M)/\ker\psi\simeq\R$. Indeed, to any class $[f]$ in the quotient we can associate the real number $\psi(f)$. This map is well-defined and injective. It is also surjective, since, for any $r\in\R$, the image of $[r\cdot 1_{\cO(M)}]$ is $r$. It follows in particular that any class in the quotient is of type $[r\cdot 1_{\cO(M)}]$ for a unique $r\in\R$. We have the following

\begin{prop}\label{Bijections}
The maps $${\frak b}:M\ni m\mapsto \zm_m\in\op{Spm}(\cO(M))$$ and $$\flat:\h_{\Z_2^n{\tt UAlg}}(\cO(M),\R)\ni\psi\mapsto\ker\psi\in\op{Spm}(\cO(M))$$ are 1:1 correspondences.
\end{prop}

\begin{proof} To prove that $\frak b$ is bijective, consider a maximal ideal $\zm\in\op{Spm}(\cO(M))$. The image $\ze_M(\zm)\subset\Ci(M)$ is a maximal ideal. Indeed, it is an ideal, since the map $\ze_M$ is surjective (the short sequence of sheafs \cite[Equation (3)]{BP1} is exact for the good reason that it is exact for any open subset of $M$). To see that it is maximal, assume there is an ideal $\zn$, such that $\ze_M(\zm)\subset\zn\subset\Ci(M)$, so that $\zm\subset\ze_M^{-1}(\zn)\subset\cO_M(M)$. It follows that $\ze_M^{-1}(\zn)=\zm$ or $\ze_M^{-1}(\zn)=\cO_M(M)$, and that $\zn=\ze_M(\zm)$ or $\zn=\Ci(M)$. Hence, $$\ze_M(\zm)=I_m=\{f\in\Ci(M):f(m)=0\}\;,$$ since any maximal ideal of $\Ci(M)$ is known to be of type $I_m$ for a unique $m\in M$. Finally, we get $$\zm\subset\ze_M^{-1}(I_m)=\{f\in\cO(M): (\ze_M f)(m)=0\}=\zm_m\subset\cO(M)\;.$$ Since $\zm_m\neq\cO(M)$, we have $\zm=\zm_m$, which proves the bijectivity of $\frak b$. Indeed, if $\zm=\zm_n$, we obtain $\ze_M(\zm_n)\subset I_n\subset \Ci(M)$, so that $$I_n=\ze_M(\zm_n)=\ze_M(\zm_m)=I_m\;,$$ and $m=n$.\medskip

Since any $\zm\in\op{Spm}(\cO(M))$ reads $\zm=\zm_m=\ker\ze_m=\flat(\ze_m)$, the map $\flat$ is surjective. Let $\psi,\phi$ be unital algebra morphisms, such that $\ker\psi=\ker\phi=\zm$. For any $f\in\cO(M)$, there exists a unique $r\in\R$, such that $[f]=[r\cdot 1_{\cO(M)}]$. Thus $\psi(f)=r=\phi(f)$ and $\psi=\phi$, so that $\flat$ is also injective. \end{proof}

\begin{prop}\label{Homeo} The map $$\ze_M^{-1}:\op{Spm}(\Ci(M))\ni I_m\mapsto \ze_M^{-1}(I_m)=\zm_m\in\op{Spm}(\cO(M))$$ is a homeomorphism with inverse $\ze_M$, both, if the maximal spectra are endowed with their Zariski topology and if they are endowed with their Gel'fand topology. Hence, the Zariski and Gel'fand topologies coincide on $\op{Spm}(\cO(M))$. Further, the bijection $${\frak b}:M\ni m\mapsto\zm_m\in\op{Spm}(\cO(M))$$ is a homeomorphism.\end{prop}

\begin{proof} The maps $$\ze_M^{-1}:\op{Spm}(\Ci(M))\ni I_m\mapsto \ze_M^{-1}(I_m)=\zm_m\in\op{Spm}(\cO(M))$$ and $$\ze_M:\op{Spm}(\cO(M))\ni\zm_m\mapsto \ze_M(\zm_m)=I_m\in\op{Spm}(\Ci(M))$$ are inverses of each other.\medskip

We first equip the spectrum $\op{Spm}(\Ci(M))$ as usual with the Zariski topology, which is defined by its basis of open subsets $V_{\Ci}(f)$, $f\in\Ci(M)$, given by $$V_{\Ci}(f)=\{I_m\in\op{Spm}(\Ci(M)):f\notin I_m\}\;,$$ and we proceed similarly for $\op{Spm}(\cO(M))$. It is straightforwardly checked that, if $f=\ze_M(F)$, we have $$\ze_M^{-1}(V_{\Ci}(f))=V_{\cO}(F)\quad \text{and} \quad \ze_M(V_{\cO}(F))=V_{\Ci}(f)\;.$$ Hence the announced homeomorphism result for the Zariski topologies.\medskip

The Gel'fand topology of $\op{Spm}(\Ci(M))$ is defined by the basis of open subsets $B_{\Ci}(m,\zve;$ $f_1,\ldots,f_n)$, indexed by $m\in M$, $\zve>0$, $n\in\N$, and $f_1,\ldots,f_n\in\Ci(M)$, and defined by $$B_{\Ci}(m,\zve;f_1,\ldots,f_n)=\{I_n\in\op{Spm}(\Ci(M)): |f_i(n)-f_i(m)|<\zve, \forall i\}\;.$$ The Gel'fand topology of $\op{Spm}(\cO(M))$ is defined analogously by $$B_{\cO}(m,\zve;F_1,\ldots,F_n)=\{\zm_n\in\op{Spm}(\cO(M)): |(\ze_MF_i)(n)-(\ze_MF_i)(m)|<\zve, \forall i\}\;,$$ where $F_i\in\cO(M)$. If $f_i=\ze_M(F_i)$, we have obviously $$\ze_M^{-1}(B_{\Ci}(m,\zve;f_1,\ldots,f_n))=B_{\cO}(m,\zve;F_1,\ldots,F_n)\;,$$ and similarly for $\ze_M$, so that the homeomorphism result holds also for the Gel'fand topologies.\medskip

Since the Zariski and Gel'fand topologies coincide on $\op{Spm}(\Ci(M))$, it follows from the above that there is a homeomorphism from $\op{Spm}(\cO(M))$ endowed with the Zariski topology to itself endowed with the Gel'fand topology.\medskip

It is well-known that the map ${\frak b}_{\Ci}:M\ni m\mapsto I_m\in\op{Spm}(\Ci(M))$ is a homeomorphism. Hence, the bijection ${\frak b}=\ze_M^{-1}\circ{\frak b}_{\Ci}$ is a homeomorphism as well.\end{proof}

\subsection{Reconstruction of the structure sheaf}

\begin{prop}\label{Invertibility}
Let $(M,\cO_M)$ be a $\Zn$-manifold and let $U\subset M$ be open. A $\Zn$-function $F\in\cO_M(U)$ is invertible if and only if its base projection $f=\ze_U(F)\in\Ci_M(U)$ is invertible.
\end{prop}

\begin{proof}
It is obvious that $f$ is invertible if $F$ is. Assume now that there exists $f^{-1}\in\Ci(U)$ and consider a cover of $U$ by $\Zn$-chart domains $V_i$. For any $i$, we have $f^{-1}|_{V_i}=(f|_{V_i})^{-1}$, i.e., the base function $\ze_{V_i}(F|_{V_i})=\ze_U(F)|_{V_i}$ is invertible in $\Ci(V_i)$, so that $F|_{V_i}\in\cO(V_i)\simeq\Ci(V_i)[[\xi]]$ has an inverse $G_{V_i}\in\cO(V_i)$. It follows that, for any $V_i$ and $V_j$ with intersection $V_{ij}$, $$G_{V_i}|_{V_{ij}}=(F|_{V_{ij}})^{-1}=G_{V_j}|_{V_{ij}}\;.$$ Hence, there is a unique $\Zn$-function $G\in\cO(U)$, such that $G|_{V_i}=G_{V_i}$. It is clear that $G$ is the inverse of $F$.
\end{proof}

Reconstructions of a sheaf from its global sections have been thoroughly studied in algebraic and differential geometry. A survey on such results can be found in \cite{BPPf} and \cite{BPPi}. The probably best known example is the construction of the structure sheaf $\cO_X$ of an affine scheme $X=\op{Spec}R$ from its global sections commutative unital ring $\cO_X(X)=R$. In this case, the ring $\cO_X(V_f)$ of functions on a Zariski open subset $V_f$, $f\in R$, is defined as a localization of $\cO_X(X)$. In the case of a $\Zn$-manifold $(M,\cO_M)$, we reconstruct $\cO_M(U)$ as the localization of $\cO_M(M)$ with respect to the multiplicative subset $S_U=\{F\in\cO_M(M):(\ze_M F)|_U\;\text{is invertible}\;\}\,.$ The chosen localization comes with a morphism that sends global sections with invertible projection in $\Ci_M(U)$ to invertible sections in $\cO_M(U)$, see Proposition \ref{Invertibility}.\medskip

We briefly address localization in the $\Zn$-commutative setting. Let $R$ be a $\Z_2^n$-commutative associative unital $\R$-algebra and let $S$ be a multiplicative subset of $R$, whose elements are homogeneous even vectors. In the following, we consider right fractions $rs^{-1}\in RS^{-1}$ (left ones would work as well) and denote the degree of any element by the same symbol as the element itself. We define the equivalence $rs^{-1}\sim r's'^{-1}$ of two fractions by requiring the existence of some $\zs\in S$ such that \be\label{Equiv}(rs'-(-1)^{\langle s,s'\rangle}r's)\,\zs=0\;.\ee It can be checked that the relation $\sim$ is an equivalence. The addition of fractions is defined by \be\label{a}rs^{-1}+r's'^{-1}=\left(rs'+(-1)^{\langle s,s'\rangle}r's\right)(ss')^{-1}\;.\ee The definition is independent of the chosen representatives. The multiplication of fractions is given by \be\label{m}rs^{-1}\cdot r's'^{-1}=(-1)^{\langle r'+s',s\rangle}\,rr'(ss')^{-1}\;,\ee provided $r'$ is homogeneous. Again the result is invariant under a change of representatives. If $r'$ is not homogeneous, it uniquely reads $r'=\sum_{\zg} r'^\zg$ ($\zg\in\Zn$). In view of \eqref{a}, we get $r's'^{-1}=\sum_\zg(r'^\zg s'^{-1})$, so that we can extend the definition \eqref{m} by linearity. Finally, the scalar multiplication by $\frak r\in\R$ is \be\label{s}{\frak r}(r s^{-1})=({\frak r}\,r)s^{-1}\;.\ee

For our purpose, it will actually be sufficient to consider the multiplicative subset $$S_U=\{F\in\cO^{\,0}_M(M):(\ze_M F)|_U\;\text{is invertible}\;\}\subset\cO_M(M)\;$$ of the $\Zn$-commutative associative unital $\R$-algebra $\cO_M(M)$. Since the $\Zn$-functions of this subset are not only even, but of degree 0, the signs in the equivalence \eqref{Equiv} and the operations \eqref{a} and \eqref{m} disappear.

\begin{prop} The operations \eqref{a}, \eqref{m}, and \eqref{s} $(\,$without signs$\,)$ endow the localization $\cO_M(M)\cdot S_U^{-1}$ with a $\Zn$-commutative associative unital $\R$-algebra structure, whose grading is naturally induced by the grading of $\cO_M(M)$ $(\,$for a homogeneous $r$, the degree of $r s^{-1}$ is the degree of $r$$\,)$, and whose zero $(\,$resp., unit$\,)$ is represented by $0\, 1^{-1}$ $(\,$resp., $1\, 1^{-1}$$\,)$.\end{prop}

\begin{proof}
Straightforward verification.
\end{proof}

We thus get a presheaf
$${\cal L}_M:{\tt Open}(M)\ni U \mapsto \cO_M(M)\cdot S_U^{-1}\in{\Zn}{\tt UAlg}$$
on $M$ valued in the category of $\Z_2^n$-commutative associative unital $\R$-algebras. Indeed, if $V\subset U$ is open, the obvious inclusion $\iota_{V}^U: S_U \hookrightarrow {S}_V$ provides a natural well-defined restriction
$$
r_V^U : {\cal L}_M(U)\ni F s^{-1}\mapsto F \, (\iota_V^U s)^{-1}\in{\cal L}_M(V)\;,
$$
and these restrictions satisfy the usual cocycle condition.\medskip

As indicated above, we will show (in several steps) that the presheaf ${\cal L}_M$ coincides with the structure sheaf $\cO_M$.\medskip

First, since it follows from Proposition \ref{Invertibility} that, for any $s\in S_U$, the restriction $s|_U$ is invertible in $\cO_M(U)$, we have a map
$$\zl_U: {\cal L}_M(U)\ni F s^{-1} \mapsto F|_U (s|_U)^{-1}\in\cO_M(U)\;.$$ This map is well-defined. Indeed, if $F s^{-1}=F' s'^{-1}$, there is $\zs\in S_U$, such that $(F|_Us'|_U-F'|_Us|_U)\,\zs|_U=0$. Since the restrictions $s|_U$, $s'|_U$, and $\zs|_U$ are invertible in $\cO_M(U)$, the claim follows. Further, it can be straightforwardly checked that $\zl_U$ is a morphism of $\Z_2^n$-graded unital $\R$-algebras.\medskip

In fact:

\begin{prop}\label{Thm:localisationamp}
For any open $U\subset M$, the localisation map $\zl_U:{\cal L}_M(U)\to\cO_M(U)$ is a $\Z_2^n$-graded unital $\R$-algebra isomorphism.
\end{prop}

The proof of this result uses a method that can be found in various works, see for instance \cite{BBH}, \cite{CCF}, and \cite{NS}. We give this proof for completeness, as well as to show that it goes through in our $\Zn$-graded stetting.

\begin{proof}
It suffices to explain why $\zl_U$ is bijective.
\begin{enumerate}
\item \emph{Injectivity:} Assume that $F|_U(s|_U)^{-1}=0$, i.e., that $F|_U=0$, and show that $F s^{-1}\sim 0 1^{-1}$, i.e., that there is $\zs\in S_U$, such that $F\zs=0$. Let $(V_i,\psi_i)$ be a partition of unity of $\cM$, such that the $V_i$ are $\Zn$-chart domains, so that $\cO_M|_{V_i}\simeq\Ci_M|_{V_i}[[\xi]]$. For any $i$, we have $$F|_{U\cap V_i}(x,\xi)=\sum_\za F_\za|_{U\cap V_i}(x)\xi^\za=0,\quad{i.e.},\quad F_\za|_{U\cap V_i}=0,\;\forall \za\;.$$ Let $$\zs_i\in\Ci_M(V_i)\subset\cO^{\,0}_M(V_i),\quad\text{such that}\quad\zs_i|_{U\cap V_i}>0\quad\text{and}\quad \zs_i|_{V_i\setminus(U\cap V_i)}=0\;.$$ It follows that $F|_{V_i}\zs_i=0$. The $\Zn$-function $\zs=\sum_i\zs_i\psi_i\in\cO_M^{\,0}(M)$ has the required properties. Indeed, the open subsets $V_i$ and $\zW_i=M\setminus\op{supp}\psi_i$ cover $M$ and $F\zs_i\psi_i$ vanishes on both, $V_i$ and $\zW_i$, so that $F\zs=\sum_iF\zs_i\psi_i=0$. In addition, for any $m\in U$, we have $(\ze\psi_i)(m)\ge 0$, for all $i$, and there is at least one $j$, such that $(\ze\psi_j)(m)>0$. Since $(\ze\psi_j)|_{\zW_j}=0$, we get $m\in U\cap V_j$, so that $(\zs_j\,\ze\psi_j)(m)>0$, $$(\ze\zs)(m)=\sum_i(\zs_i\,\ze\psi_i)(m)>0\;,$$ and $\zs\in S_U$.

\item \emph{Surjectivity:} We must express an arbitrary $f \in \cO_M(U)$ as a product $f = F|_U(s|_U)^{-1}$, with $F\in\cO_M(M)$ and $s\in S_U$. To construct the global sections $F$ and $s$, consider an increasing countable family of seminorms $p_n$ that implements the locally convex topology of the Fr\'echet space $\cO_M(M)$. Take also a countable open cover $U_n$ of $U$, such that $\bar U_n\subset U$, as well as bump functions $\zg_n\in\cO_M^{\,0}(M)$, which satisfy $\zg_n|_{U_n}=1$ (which implies that $(\ze_M\zg_n)|_{U_n}=1$), $\op{supp}\zg_n\subset U$, and $\ze_M\zg_n\ge 0$. The following series converge in $\cO_M(M)$ and provide us with the required global sections:
\begin{align*}
& F := \sum_{n=0}^\infty \frac{1}{2^n} \frac{\zg_n \: f}{1+ p_n(\zg_n) + p_n(\zg_n \: f)}\quad \text{and}\quad s := \sum_{n=0}^\infty \frac{1}{2^n} \frac{\zg_n}{1+ p_n(\zg_n) + p_n(\zg_n \: f)}\;.
\end{align*}
\vspace{1cm}Indeed, convergence follows, if we can show that the series are Cauchy, i.e., if they are Cauchy with respect to each $p_m$. If $r,s\to\infty$, we get $r\ge m$, and, since the seminorms are increasing, we have
$$p_m \left( \sum_{n=r}^s \frac{1}{2^n} \frac{\zg_n \: f}{1+ p_n(\zg_n) + p_n(\zg_n \: f)}\right) \leq  \sum _{n=r}^s \frac{1}{2^n} \frac{p_m(\zg_n \: f)}{1+ p_n(\zg_n) + p_n(\zg_n \: f)} < \sum_{n=r}^s \frac{1}{2^n}\to 0,$$ whether the factor $f$ is present or not. On the other hand, as restrictions are continuous, it is clear that $F|_U = f s|_U$, so that $f=F|_U(s|_U)^{-1}$, provided we show that $s\in\cO_M^{\,0}(M)$ belongs to $S_U$, i.e., that $(\ze_M s)(m)\neq 0$, for all $m\in U$. Too see this, remark that $(\id,\ze):(M,\Ci_M)\to (M,\cO_M)$ is a morphism of $\Zn$-manifolds, so that $\ze_M:\cO_M(M)\to\Ci_M(M)$ is continuous, see \cite[Theorem 19]{BP1}. For any $U_m$ of the cover of $U$, we thus get $$(\ze_M s)|_{U_m} := \sum_{n=0}^\infty \frac{1}{2^n} \frac{(\ze_M\zg_n)|_{U_m}}{1+ p_n(\zg_n) + p_n(\zg_n \: f)}>0\;,$$ in view of the properties of $\zg_n$.\end{enumerate}\end{proof}
\begin{theorem}\label{Thm:localisationofsheaf}
The $\Z_2^n$-commutative associative unital $\R$-algebra $\cO_M(M)$ of global sections of the structure sheaf of a $\Zn$-manifold $(M, \cO_M)$ fully determines this sheaf. More precisely, there is a presheaf isomorphism $\zl:{\cal L}_M\to \cO_M$, so that the presheaf ${\cal L}_M$, which is obtained from $\cO_M(M)$, is actually a sheaf, which is isomorphic to the structure sheaf $\cO_M$.
\end{theorem}
\begin{proof} It suffices to check that the family $\zl_U:{\cal L}_M(U)\to\cO_M(U)$, $U\in{\tt Open}(M)$, of $\Zn$-graded unital $\R$-algebra isomorphisms, commutes with the restrictions $r^U_V$ in ${\cal L}_M$ and $\zr^U_V$ in $\cO_M$ ($V\subset U$, $V\in{\tt Open}(M)$). This is actually obvious: $$\zl_V(r^U_V(F s^{-1}))=F|_V(s|_V)^{-1}=\zr^U_V(\zl_U(F s^{-1}))\;.$$\end{proof}

\subsection{Reconstruction of a $\Zn$-morphism}

In algebraic geometry, any commutative unital ring morphism $\psi:S\to R$ defines a morphism of affine schemes $\Phi=(\phi,\phi^*):(\op{Spec}R,\cO_{\op{Spec}R})\to (\op{Spec}S,\cO_{\op{Spec}S})$, whose continuous base map $\phi$ associates to each prime ideal $\frak p$ the prime ideal $\psi^{-1}({\frak p})$. A similar result exists in the category of $\Zn$-manifolds and $\Zn$-morphisms, with the same definition of the continuous base map.

\begin{theorem}\label{Thm:SpaceAlgebraMorphBijection}
Let $\mathcal{M} = (M, \cO_M)$ and $\,\mathcal{N} = (N, \cO_N)$  be $\Z_2^n$-manifolds. The map
$$\zb:\Hom_{\Z_2^n{\tt Man}}\big(\mathcal{M},\mathcal{N} \big)\ni\Phi=(\phi,\phi^*) \mapsto \phi^*_N\in\Hom_{\Z_2^n{\tt UAlg}}\big ( \cO_N(N), \cO_M(M) \big)$$
is a bijection.
\end{theorem}
\begin{proof}
To show that $\zb$ is surjective, we consider $\psi\in\h_{\Zn\tt UAlg}(\cO_N(N),\cO_M(M))$ and construct $\Phi\in\h_{\Zn\tt Man}(\cM,\cN)$, such that $\phi^*_N=\psi$.\medskip

Since $M$ (resp., $N$) endowed with its base space topology is homeomorphic to $\op{Spm}(\cO_M(M))$ (resp., $\op{Spm}(\cO_N(N))$) endowed with the Zariski topology, we define $\zf$ by $$\zf:\op{Spm}(\cO_M(M))\ni\ker\ze_m\mapsto \ker(\ze_m\circ\psi)=\ker\ze_n\in\op{Spm}(\cO_N(N))\;,$$ see Propositions \ref{Homeo} and \ref{Bijections}. This map is continuous. Indeed, for any $F\in\cO_N(N)$, the preimage by $\zf$ of the open subset $$V(F)=\{\ker\ze_n\in\op{Spm}(\cO_N(N)): \ze_n(F)\neq 0\}$$ is the subset $$\zf^{-1}(V(F))=\{\ker\ze_m\in\op{Spm}(\cO_M(M)):\ze_m(\psi(F))\neq 0\}=V(\psi(F))\;.$$ 

\newcommand{\cL}{{\cal L}}

To define, for any open $V\subset N$, a $\Zn$-graded unital $\R$-algebra morphism $$\phi^*_V:\cO_N(V)\to (\phi_*\cO_M)(V)\;,$$ we rely on the isomorphism of $\Zn$-graded unital $\R$-algebras $\cO_N(V)\simeq \cL_N(V)$ and the similar isomorphism in $M$. Hence, we define $\phi^*_V$ by $$\phi^*_V:\cL_N(V)\ni F s^{-1}\mapsto \psi(F)\psi(s)^{-1}\in\cL_M(\phi^{-1}(V))\;.$$ This map is actually well-defined. Since $s\notin\ker\ze_n$, for all $n\in V$, we have $s\notin \ker(\ze_m\circ \psi),$ for all $m\in\phi^{-1}(V)$, what means that $\psi(s)\in S_{\zf^{-1}(V)}$. In view of this, it is easy to see that the image is independent of the representative. The map $\phi^*_V$ is a $\Zn$-graded unital $\R$-algebra morphism, because $\psi$ is.\medskip

As the family $\zf^*_V$, $V\in{\tt Open}(N)$, commutes obviously with restrictions, the continuous base map $\zf$ and the family of algebra morphisms $\zf^*_V$, $V\in{\tt Open}(N)$, define a $\Zn$-morphism $\zF:\cM\to\cN\,$. Too see that $\zb(\Phi)=\zf^*_N=\psi$, it suffices to note that $\zl_N:\cL_N(N)\to\cO_N(N)$ sends the fraction $Fs^{-1}$ to the section $Fs^{-1}$ (and similarly for $M$), so that $\zf^*_N$ and $\psi$ coincide.\medskip

It remains to prove that $\zb$ is injective. Let thus $\zF=(\zf,\zf^*)$ and $\Psi=(\psi,\psi^*)$ be two $\Zn$-morphisms from $\cM$ to $\cN$, such that $\zf^*_N=\psi^*_N$. Since the pullbacks by a $\Zn$-morphism commute with the base projections, we get, for any $m\in M$, $$\zf(m)\simeq\ker\ze_{\zf(m)}=\{F\in\cO(N):((\ze_NF)\circ\zf)(m)=0\}=\{F\in\cO(N):(\ze_M(\zf^*_NF))(m)=0\}\;.$$ Hence, the continuous base maps $\zf$ and $\psi$ coincide. Similarly, for any open $V\subset N$, each $F_V\in\cO_N(V)$ reads uniquely $F_V=\zl_V(Fs^{-1})=F|_V(s|_V)^{-1}$, with $F\in\cO_N(N)$ and $s\in S_V\subset\cO_N(N)$, see Proposition \ref{Thm:localisationamp}. As the family of pullbacks $\zf^*$ commutes with restrictions, we obtain $$\zf^*_V(F_V)=(\zf^*_VF|_V)\,(\zf^*_Vs|_V)^{-1}=(\zf^*_NF)|_{\zf^{-1}(V)}\,(\zf^*_Ns)|_{\zf^{-1}(V)}^{-1}\;.$$ Hence $\zf^*_V=\psi^*_V$.\end{proof}

The preceding theorem, which allows us to characterize $\cM$-points $\h_{\Zn\tt Man}(\cM,\cN)$ of a $\Zn$-manifold $\cN$ by algebra morphisms, has some noteworthy corollaries.

\begin{cor}\label{Thm:FunctorManAlg}
The covariant functor
$${\cF} : \Z_2^n{\tt Man} \to \Z_2^n{\tt UAlg}^{\tt op}\,,$$
which is defined on objects by $\cF(\cM)=\cO_M(M)$ and on morphisms by $\cF(\zF)=\zf^*_N$, is fully faithful, so that $\Zn${\tt Man} can be viewed as full subcategory of $\Zn{\tt UAlg}^{\op{op}}$.
\end{cor}

The statement regarding the full subcategory is based on the well-known fact that any fully faithful functor is injective up to isomorphism on objects. This means that the existence of an isomorphism $\cO_M(M)\simeq\cO_N(N)$ of $\Zn$-graded unital $\R$-algebras implies the existence of an isomorphism $\cM\simeq\cN$ of $\Zn$-manifolds.

\begin{cor}\label{Pursell-Shanks}
Let $\mathcal{M}=(M ,\cO_M)$ and $\mathcal{N}=(N , \cO_N)$ be $\Z_2^n$-manifolds. The $\Zn$-manifolds $\mathcal{M}$ and $\mathcal{N}$ are diffeomorphic if and only if their $\Z_2^n$-commutative associative unital $\R$-algebras $\cO_M(M)$ and $\cO_N(N)$ of global $\Zn$-functions are isomorphic.
\end{cor}

Such Pursell-Shanks type results have been studied extensively by one of the authors of this paper. Algebraic characterizations similar to Corollary \ref{Pursell-Shanks} exist for instance for the Lie algebras of first order differential operators, of differential operators, and symbols of differential operators on a smooth manifold, for the super Lie algebras of vector fields and first order differential operators on a smooth supermanifold, as well as for the Lie algebra of sections of an Atiyah algebroid, see \cite{GP1}, \cite{GP3}, \cite{GP4}, \cite{GP5}.

\begin{cor}\label{IniTer}
The $\Zn$-manifold $\cE=(\emptyset,0)$ $(\,$resp., $\R^{0|\mathbf{0}}=(\{\op{pt}\},\R)$$\,)$ is the initial $(\,$resp., terminal$\,)$ object of the category of $\Zn$-manifolds.
\end{cor}

\begin{proof}
For any $\Zn$-manifold $\cM=(M,\cO_M)$, we have bijections $$\h_{\Zn\tt Man}(\cE,\cM)\simeq \h_{\Zn{\tt UAlg}}(\cO_M(M),0)\simeq \{F\mapsto 0\}\;$$ and $$\h_{\Zn\tt Man}(\cM,\R^{0|\mathbf{0}})\simeq \h_{\Zn{\tt UAlg}}(\R,\cO_M(M))\simeq\{r\mapsto r\cdot 1\}\;.$$\end{proof}

\section{Finite products in the category of $\Zn$-manifolds}

\subsection{Cartesian product of $\Z_2^n$-manifolds}\label{CartProdZ2nMan}

Let $\cM=(M,\cO_M)$ and ${\cN}=(N,\cO_N)$ be two $\Zn$-manifolds of dimension $p|\mathbf{q}$ and $r|\mathbf{s}\,$, respectively. The products $U\times V$, $U\subset M$ and $V\subset N$ open, form a basis $\cB$ of the (second-countable, Hausdorff) product topology of $M\times N$. Better, since the $\Zn$-chart domains $U_i$ in $M$ (resp., $V_j$ in $N$) are a basis of the topology of $M$ (resp., of $N$), the products $U_i\times V_j$ form a basis $\frak B$ of the product topology of $M\times N$. As $\Z_2^n$-chart domains are diffeomorphic to open subsets of some coordinate space $\R^n$, we identify the $U_i$ and the $V_j$ with the diffeomorphic $U_i\subset\R^p$ and $V_j\subset\R^r$. Further, we denote the coordinates of the charts with domains $U_i$ (resp., $V_j$) by $(x_i,\xi_i)$ (resp., $(y_j,\zh_j)$), or, in case we use only two domains $U_i$ (resp., $V_j$), we write also $(x,\xi)$ and $(x',\xi')$ (resp., $(y,\zh)$ and $(y',\zh')$).\medskip

\begin{defi}\label{ProdZnMan} Let $\cM=(M,\cO_M)$ and ${\cN}=(N,\cO_N)$ be two $\Zn$-manifolds of dimension $p|\mathbf{q}$ and $r|\mathbf{s}\,$, respectively. The \emph{product $\Zn$-manifold} $\cM\times\cN$, of dimension $p+r|\mathbf{q}+\mathbf{s}$, is the locally $\Zn$-ringed space $(M\times N,\cO_{M\times N})$, where $M\times N$ is the product topological space and where the sheaf $\cO_{M\times N}$ is glued from the sheaves $\Ci_{U_i\times V_j}(x_i,y_j)[[\xi_i,\zh_j]]$ associated to the basis $\frak B$: \be\cO_{M\times N}|_{U_i\times V_j}\simeq\Ci_{U_i\times V_j}(x_i,y_j)[[\xi_i,\zh_j]]\;.\label{DefProdMan}\ee\end{defi}

Recall that sheaves can be glued. More precisely, if $(U_i)_i$ is an open cover of a topological space $M$, if ${\cal F}_i$ is a sheaf on $U_i$, and if $\zvf_{ji}:{\cal F}_i|_{U_i\cap U_j}\to {\cal F}_j|_{U_i\cap U_j}$ is a sheaf isomorphism such that the usual cocycle condition $\zvf_{kj}\,\zvf_{ji}=\zvf_{ki}$ holds, then there is a unique sheaf ${\cal F}$ on $M$ such that ${\cal F}|_{U_i}\simeq{\cal F}_i$. In the following, we set $U_{ij}=U_i\cap U_j$.\medskip

Let now $\Ci_{U_i\times V_j}[[\xi,\zh]]$ be the standard sheaf of $\Zn$-graded $\Zn$-commutative associative unital $\R$-algebras of formal power series in $(\xi,\zh)$ with coefficients in sections of the sheaf $\Ci_{U_i\times V_j}$. The isomorphisms $\zvf_{\frak{ij},ij}$ between the appropriate restrictions of the sheaves of algebras $\Ci_{U_i\times V_j}[[\xi,\zh]]$ on the open cover $(U_i\times V_j)_{i,j}$ of $M\times N$ are induced as follows. Since $\cM$ is a $\Zn$-manifold, there are $\Zn$-isomorphisms $$\Phi_i=(\zf_i,\zf_i^*):(U_i,\cO_M|_{U_i})\to(U_i,\Ci_{U_i}[[\xi]])\;,$$ which induce $\Zn$-isomorphisms or coordinate transformations $$\Psi_{\frak{i}i}=\Phi_{\frak i}\Phi_i^{-1}:(U_{i\frak{i}},\Ci_{U_i}|_{U_{i\frak{i}}}[[\xi]])\to (U_{i\frak{i}},\Ci_{U_{\frak i}}|_{U_{i\frak{i}}}[[\xi']])\;.$$ As we view $U_i$ as both, an open subset of $M$ and an open subset of $\R^p$, we implicitly identify $U_i$ with its diffeomorphic image $\zf_i(U_i)$, so that $\zf_i=\id_{U_i}$. Hence, the coordinate transformations reduce to the isomorphisms \be\label{CocycleTrivial}\psi^*_{{\frak i}i}=(\zf_i^*)^{-1}\zf_{\frak i}^*\ee of sheaves of $\Zn$-commutative $\R$-algebras: \be\label{CTU}\psi^*_{{\frak i}i}:\Ci_{U_{\frak{i}}}|_{U_{i{\frak{i}}}}[[\xi']]\to \Ci_{U_i}|_{U_{i{\frak{i}}}}[[\xi]]\;.\ee Similar coordinate transformations exist for $\cN$: \be\label{CTV}\psi^*_{\frak{j}j}:\Ci_{V_{\frak{j}}}|_{V_{j{\frak{j}}}}[[\zh']]\to \Ci_{V_j}|_{V_{j{\frak{j}}}}[[\zh]]\;.\ee We denote the base coordinates in $U_i$ (resp., $U_{\frak i}$) by $x$ (resp., $x'$) and those in $V_j$ (resp., $V_{\frak j}$) by $y$ (resp., $y'$). The coordinate transformations \eqref{CTU}, $x=x(x',\xi'), \xi=\xi(x',\xi')$, and \eqref{CTV}, $y=y(y',\zh'), \zh=\zh(y',\zh')$, implement coordinate transformations or isomorphisms of sheaves of $\Zn$-commutative $\R$-algebras $$\zvf_{\frak{ij},ij}=\psi^*_{{\frak i}i}\times\psi^*_{\frak{j}j}:\Ci_{U_{\frak i}\times V_{\frak j}}|_{U_{i\frak i}\times V_{j\frak j}}[[\xi',\zh']]\to \Ci_{U_i\times V_j}|_{U_{i\frak i}\times V_{j\frak j}}[[\xi,\zh]]\;.$$ In view of \eqref{CocycleTrivial}, the $\zvf_{\frak{ij},ij}$ satisfy the cocycle condition. We thus get a unique glued sheaf $\cO_{M\times N}$ of $\Zn$-commutative $\R$-algebras over $M\times N$ which restricts on $U_i\times V_j$ to $$\cO_{M\times N}|_{U_i\times V_j}\simeq \Ci_{U_i\times V_j}[[\xi,\zh]]\;,$$ i.e., we obtain a $\Zn$-manifold, which we refer to as the product $\cM\times\cN$ of $\cM$ and $\cN$.

\subsection{Fundamental isomorphisms}

\begin{theo}\label{FundaIsom} Let $\,\R^{p|\mathbf{q}}$ $(\,$resp., $\R^{r|\mathbf{s}}$$\,)$ be the usual $\Zn$-domain $(\R^p,\Ci_{\R^p}[[\xi]])$ $(\,$resp., $(\R^r,\Ci_{\R^r}$ $[[\zh]])$$\,)$, and let $\zW'\subset\R^p$ and $\zW''\subset\R^r$ be open. There is an isomorphism of topological algebras \be\label{CompProdZ2n} \Ci_{\R^p}(\zW')[[\xi]]\,\widehat\0\;\Ci_{\R^r}(\zW'')[[\zh]]\simeq\Ci_{\R^p\times\R^r}(\zW'\times \zW'')[[\xi,\zh]]\;,\ee where the completion is taken with respect to any locally convex topology on the algebraic tensor product $\Ci_{\R^p}(\zW')[[\xi]]\0\Ci_{\R^r}(\zW'')[[\zh]]$. \end{theo}

\begin{proof} Let $R$ be a commutative von Neumann regular ring. For any families $(M_\za)_\za$ and $(N_\zb)_\zb$ of free $R$-modules, the natural $R$-linear map $$(\prod_\za M_\za)\0_R(\prod_\zb N_\zb)\to \prod_{\za\zb}(M_\za\0_R N_\zb)$$ is injective, if and only if $R$ is injective as a module over itself \cite{Good}. Since any field is von Neumann regular, the regularity and injective module conditions are satisfied for $R=\R$. Hence, the linear map $$(\prod_\za \Ci(\zW'))\0(\prod_\zb \Ci(\zW''))\to \prod_{\za\zb}(\Ci(\zW')\0 \Ci(\zW''))$$ is injective. Further, in view of \cite[Corollary 17]{BP1}, the map \be\label{TVSIsom2}\Ci(\zW')[[\xi]]\ni\sum_{\za\in \cA} f_\za(x)\xi^\za\mapsto (f_\za)_{\za\in \cA}\in\prod_{\za\in \cA}\Ci(\zW')\ee is a {\small TVS}-isomorphism between the source and the target equipped with the standard topology and the product topology of the standard topologies, respectively. In the sequence of canonical maps $$\Ci(\zW')[[\xi]]\0\Ci(\zW'')[[\zh]]\simeq (\prod_{\za\in \cA} \Ci(\zW'))\0(\prod_{\zb\in\cB}\Ci(\zW''))\to\prod_{\za\in\cA\,,\,\zb\in\cB}(\Ci(\zW')\0 \Ci(\zW''))\to$$ \be\label{SEA}\prod_{\za\in\cA\,,\,\zb\in\cB}\Ci(\zW'\times \zW'')\simeq\Ci(\zW'\times \zW'')[[\xi,\zh]]\;,\ee the first $\simeq$ is a linear bijection, the first $\to$ is a linear injection, and the second $\simeq$ is a {\small TVS}-isomorphism for the topologies used in \eqref{TVSIsom2}. In \be\label{Bas}\Ci(\zW')\0 \Ci(\zW'')\to \Ci(\zW')\widehat\0\, \Ci(\zW'')\simeq\Ci(\zW'\times \zW'')\;,\ee the isomorphism $\simeq$ is the well-known {\small TVS}-isomorphism \cite{Gro} (the target is endowed with its standard topology and the source with the topology of the completion with respect to any ($\Ci(\zW')$ is nuclear) locally convex topology on $\Ci(\zW')\0\Ci(C)$ -- we will not specify the latter topology), and the arrow $\to$ is the continuous linear inclusion (any {\small TVS} is a topological vector subspace ({\small TVSS}) of its completion, see Proposition \ref{Completion}). This $\to$ induces the second $\to$ in \eqref{SEA}, which is the inclusion of the source vector subspace into the target vector space. The source becomes a {\small TVSS} of the target when endowed with the induced topology (the induced topology is coarser than the product topology of the induced topologies). Finally, we equip the first space in \eqref{SEA} with the initial topology with respect to the first $\to\,,$ so that the first space gets promoted to a {\small TVSS} of the second, see Proposition \ref{InjLinIniTop}, and the first $\to$ becomes the continuous linear inclusion. The composite $$\imath:\Ci(\zW')[[\xi]]\0\Ci(\zW'')[[\zh]]\ni\sum_\za f_\za\,\xi^\za\0\sum_\zb g_\zb\,\zh^\zb\mapsto \sum_{\za\zb}f_\za\0 g_\zb\,\xi^\za\zh^\zb\in\Ci(\zW'\times \zW'')[[\xi,\zh]]\;$$ of the maps of \eqref{SEA} is now the inclusion of a the source {\small TVSS} into the target {\small TVS}.\medskip

Note that, since the target is a {\small LCTVS}, see \cite[Lemma 16]{BP1}, the source {\small TVSS} is also a {\small LCTVS}, see Proposition \ref{IniTopLC} (since $\Ci(\zW')[[\xi]]$ is nuclear, see \cite[Lemma 16]{BP1}, the completion of the source is independent of the chosen locally convex topology). In view of Proposition \ref{CompletionTVSS}, the completion of the source is a {\small TVSS} of the completion of the target, which, as the target is complete, see \cite[Lemma 16]{BP1}, can be identified with the target due to Remark \ref{CompletionComplete}. In other words, the continuous extension \be\label{ContExt}\hat{\imath}:\Ci(\zW')[[\xi]]\widehat{\0}\,\Ci(\zW'')[[\zh]]\to\Ci(\zW'\times \zW'')[[\xi,\zh]]\;\ee of the inclusion $\imath$ is an injective continuous linear map, see text above Proposition \ref{CompletionTVSS}. We will now prove that this map is surjective.\medskip

Let $$S=\sum_{\za\zb}F_{\za\zb}\,\xi^\za\zh^\zb\in\Ci(\zW'\times \zW'')[[\xi,\zh]]$$ be a formal series in the target space. In view of \eqref{Bas}, we have \cite{Gro}, for any $(\za,\zb)\in\cA\times\cB$, $$F_{\za\zb}=\lim_{N\to+\infty}\sum_{j=0}^N f_{\za\zb}^j\0\,g_{\za\zb}^j\;,$$ where $f_{\za\zb}^j\in\Ci(\zW')$ and $g_{\za\zb}^j\in\Ci(\zW'')$, and where the limit is taken in $\Ci(\zW'\times\zW'')$. Recall that $\cA=\N^{\times |\mathbf{q'}|}\times\Z_2^{\times |\mathbf{q''}|}$, and similarly for $\cB$. The product $\cA\times\cB$ is countable, since it is a finite product of countable sets. Let $I:\cA\times\cB\to\N$ be an injective map valued in $\N$. The map $J:\cA\times\cB\to \cI$, with $\cI=I(\cA\times\cB)$, is thus a 1:1 correspondence. We identify $\cA\times\cB$ with $\cI$ via $J$. For any $j\in\N$, we set,
\begin{itemize}
\item for any $\za\in\cA$ and any $i\in\cI$, $$\Ci(\zW')\ni\zf_{\za i}^j=\begin{cases}0,\;\text{if}\;i\simeq(\zg,\zd)\neq(\za,\zd)\;,\\ f_{\za\zd}^j,\;\text{if}\;i\simeq(\zg,\zd)=(\za,\zd)\;,\end{cases}\quad\text{and}\,,$$
\item for any $\zb\in\cB$ and any $i\in\cI$, $$\Ci(\zW'')\ni\psi_{i\zb}^j=\begin{cases}0,\;\text{if}\;i\simeq(\zg,\zd)\neq(\zg,\zb)\;,\\ g_{\zg\zb}^j,\;\text{if}\;i\simeq(\zg,\zd)=(\zg,\zb)\;.\end{cases}$$
\end{itemize}
Note that $\cI$ is a finite set $\{0,1,\ldots,L\}$, $L\in\N$, (resp., is $\N$), if $\cA\times\cB$ is finite (resp., if $\cA\times\cB$ is countably infinite). For all $j\in\N$ and all $(\za,\zb)\in\cA\times\cB$, we get $$\sum_{i=0}^M\zf^j_{\za i}\0\psi^j_{i\zb}=f^j_{\za\zb}\0 g^j_{\za\zb}\;,$$ when $M\in\cI\,\cap\, [J(\za,\zb),+\infty[$. Indeed, if $i\simeq(\zg,\zd)\neq (\za,\zb)$, then, either $\zg\neq\za$ and $\zf^j_{\za i}=0$, or $\zd\neq\zb$ and $\psi^j_{i\zb}=0$. However, if $i\simeq(\zg,\zd)=(\za,\zb)$, then $\zf^j_{\za i}=f^j_{\za\zb}$ and $\psi^j_{i\zb}=g^j_{\za\zb}$, so that the announced result follows. Hence, for any $j\in\N$ and any $(\za,\zb)\in\cA\times\cB$, we have $$\lim_{M\to+\infty}\sum_{i=0}^M\zf^j_{\za i}\0\psi^j_{i\zb}= f^j_{\za\zb}\0 g^j_{\za\zb}\;,$$ where the sequence is constant for $M\ge J(\za,\zb)$ and where the limit is computed in the topology of $\Ci(\zW'\times\zW'')$. If a finite number of sequences of a {\small TVS} do converge, then their sum converges to the sum of the limits. It follows that, for any $(\za,\zb)\in\cA\times\cB$ and any $N\in\N$,
$$\lim_{M\to+\infty}\sum_{j=0}^N\sum_{i=0}^M\zf^j_{\za i}\0\psi^j_{i\zb}=\sum_{j=0}^Nf^j_{\za\zb}\0 g^j_{\za\zb}\;,$$ so that, for all $(\za,\zb)\in\cA\times\cB$,  $$\lim_{N\to+\infty}\lim_{M\to+\infty}\sum_{j=0}^N\sum_{i=0}^M\zf^j_{\za i}\0\psi^j_{i\zb}=\lim_{N\to+\infty}\sum_{j=0}^Nf^j_{\za\zb}\0 g^j_{\za\zb}=F_{\za\zb}\;$$ in $\Ci(\zW'\times\zW'')$, and
\be\label{TargetConv}\lim_{N\to+\infty}\lim_{M\to+\infty}\left(\sum_{j=0}^N\sum_{i=0}^M\zf^j_{\za i}\0\psi^j_{i\zb}\right)_{(\za,\zb)\in\cA\times\cB}=\left(F_{\za\zb}\right)_{(\za,\zb)\in\cA\times\cB}\;\ee in the product topology of $\prod_{\za\zb}\Ci(\zW'\times\zW'')$, i.e., in the topology of the {\small TVS} $\Ci(\zW'\times\zW'')[[\xi,\zh]]$. Therefore, the sequence
$$
\left(\sum_{j=0}^N\sum_{i=0}^M\zf^j_{\za i}\0\psi^j_{i\zb}\right)_{(\za,\zb)}=\sum_{j=0}^N\sum_{i=0}^M\left(\sum_\za\zf^j_{\za i}\,\xi^\za\0\sum_\zb \psi^j_{i\zb}\,\zh^\zb\right)\in\Ci(\zW')[[\xi]]\0\Ci(\zW'')[[\zh]]
$$
is a Cauchy sequence in $\Ci(\zW'\times\zW'')[[\xi,\zh]]$, so a Cauchy sequence in the {\small TVSS} $\Ci(\zW')[[\xi]]\0\,\Ci(\zW'')[[\zh]]$, and also in the topological vector supspace $\Ci(\zW')[[\xi]]\widehat\0\,\Ci(\zW'')[[\zh]]$. Since this completion is sequentially complete, the Cauchy sequence considered converges in this space: $$\lim_{N\to+\infty}\lim_{M\to+\infty}\sum_{j=0}^N\sum_{i=0}^M\left(\sum_\za\zf^j_{\za i}\,\xi^\za\0\sum_\zb \psi^j_{i\zb}\,\zh^\zb\right)\in\Ci(\zW')[[\xi]]\widehat\0\,\Ci(\zW'')[[\zh]]\;,$$ where the limit is taken in the topology of $\Ci(\zW')[[\xi]]\widehat\0\,\Ci(\zW'')[[\zh]]$. Since the inclusion $\hat\imath$, see Equation \eqref{ContExt}, is sequentially continuous, we get $$\hat\imath\;\lim_{N\to+\infty}\lim_{M\to+\infty}\sum_{j=0}^N\sum_{i=0}^M\left(\sum_\za\zf^j_{\za i}\,\xi^\za\0\sum_\zb \psi^j_{i\zb}\,\zh^\zb\right)=$$ $$\lim_{N\to+\infty}\lim_{M\to+\infty}\hat\imath\;\sum_{j=0}^N\sum_{i=0}^M\left(\sum_\za\zf^j_{\za i}\,\xi^\za\0\sum_\zb \psi^j_{i\zb}\,\zh^\zb\right)=$$
$$\lim_{N\to+\infty}\lim_{M\to+\infty}\sum_{\za\zb}\,\sum_{j=0}^N\sum_{i=0}^M\zf^j_{\za i}\0\psi^j_{i\zb}\,\xi^\za\zh^\zb=\sum_{\za\zb}F_{\za\zb}\,\xi^\za\zh^\zb=S\;,$$ in view of \eqref{TargetConv}. This shows that the continuous linear inclusion $$\hat\imath:\Ci(\zW')[[\xi]]\widehat\0\,\Ci(\zW'')[[\zh]]\to\Ci(\zW'\times\zW'')[[\xi,\zh]]$$ is bijective, so that the source {\small TVSS} of the target coincides with the target as {\small TVS}.\medskip

Since the completed tensor product of two nuclear Fr\'echet algebras is again a nuclear Fr\'echet algebra \cite[Lemma 1.2.13]{Emerton}, the source and target are actually topological algebras. We leave it to the reader to check that the preceding identification respects the multiplications.\end{proof}

\begin{rem} If $p|\mathbf{q}=p|\mathbf{0}$ and $r|\mathbf{s}=0|\mathbf{s}$, it follows from Theorem \ref{FundaIsom} that $$\Ci(\zW)\widehat\0_\R\,\R[[\xi]]\simeq\Ci(\zW)[[\xi]]\;,$$ and, if $p|\mathbf{q}=0|\mathbf{q}$ and $r|\mathbf{s}=0|\mathbf{s}$, we get $$\R[[\xi]]\widehat\0_\R\,\R[[\zh]]\simeq\R[[\xi,\zh]]\;.$$ Conversely, the general isomorphism of Theorem \ref{FundaIsom} is a consequence of the preceding particular cases and the fact that the category of complete nuclear spaces is a symmetric monoidal category with respect to the completed tensor product \emph{\cite{Cos}}.\end{rem}

\begin{theo}\label{IsomExtShfifi} There is an isomorphism of sheaves of $\Zn$-commutative $\R$-algebras $$\cO_{M\times N}\simeq (\cO_M\widehat\0\,\cO_N)^{-\,+}\;$$ between the structure sheaf of a product $\Zn$-manifold and the sheafification of the standard extension of the $\cB$-presheaf $$\cO_M\widehat\0\,\cO_N:U\times V\mapsto\cO_M(U)\widehat\0\,\cO_N(V)\;.$$  \end{theo}

\begin{proof} Recall that $\cB$ (resp., $\frak B$) is the basis of the product topology of $M\times N$ made of the rectangular subsets $U\times V$, where $U\subset M$ and $V\subset N$ are open (resp., of the rectangular subsets $U_i\times V_j$, where $U_i\subset M$ and $V_j\subset N$ are $\Zn$-chart domains). Let \be\label{LocIsom1}\cF(U\times V):=\cO_M(U)\widehat\0\,\cO_N(V)\ee be the completed tensor product of the nuclear $\Zn$-graded Fr\'echet algebras $\cO_M(U)$ and $\cO_N(V)$ (with respect to any (reasonable) locally convex topology, e.g., the projective one). If $U'\times V'\subset U\times V$, the restrictions $\zr^U_{U'}:\cO_M(U)\to\cO_M(U')$ and $\zr^V_{V'}:\cO_N(V)\to\cO_N(V')$ of the Fr\'echet sheaves $\cO_M$ and $\cO_N$ are continuous linear maps. The continuous extension of the continuous linear map $\zr^U_{U'}\0\zr^V_{V'}$ is a continuous linear map \cite{Gro} $$\zr^U_{U'}\widehat\0\,\zr^V_{V'}:\cO_M(U)\widehat\0\,\cO_N(V)\to \cO_M(U')\widehat\0\,\cO_N(V')\;,$$ which we denote by $$\zr^{U\times V}_{U'\times V'}:\cF(U\times V)\to\cF(U'\times V')\;.$$ Since the $\zr^U_{U'}$ and the $\zr^V_{V'}$ satisfy the standard presheaf conditions and the linear maps $\zr^{U\times V}_{U'\times V'}$ are continuous, it is clear that the latter satisfy these conditions as well. Hence, the pair $(\cF,\zr)$ is a $\tt Set$-valued $\cB$-presheaf.\medskip

This $\cB$-presheaf can be extended to a $\tt Set$-valued presheaf $(\bar\cF,\bar\zr)$. Indeed, set, for any open $\zW\subset M\times N$, \be\label{LocSheaf1}\bar\cF(\zW):=\{(f_{ab})_{ab}:f_{ab}\in\cF(U_a\times V_b), U_a\times V_b\subset\zW\;,\text{such that}\;\zr^{U_a\times V_b}_{U_{a\frak a}\times V_{b\frak b}}(f_{ab})=\zr^{U_{\frak a}\times V_{\frak b}}_{U_{a\frak a}\times V_{b\frak b}}(f_{\frak{ab}})\}\;,\ee and consider, for any $\zW'\subset\zW$, the map $$\bar\zr\,^\zW_{\zW'}:\bar\cF(\zW)\to\bar\cF(\zW')$$ which sends any element of $\bar\cF(\zW)$ to the element of $\bar\cF(\zW')$ that we obtain by suppressing the $f_{ab}$ for which $U_a\times V_b$ is not a subset of $\zW'$. The $\bar\zr\,^\zW_{\zW'}$ satisfy of course the standard presheaf conditions. Further, the presheaf $(\bar\cF,\bar\zr)$ extends the $\cB$-presheaf $(\cF,\zr)$. Indeed, for $\zW=U\times V$, any $f\in\cF(\zW)$ provides a unique family $f_{ab}=\zr^{U\times V}_{U_a\times V_b}(f)$ in $\bar\cF(\zW)$, thus defining a map $\flat_\zW:\cF(\zW)\to \bar\cF(\zW)$. If $\flat_\zW(f)=\flat_\zW(g)$, then, in particular, $$f=\zr^{U\times V}_{U\times V}(f)=\zr^{U\times V}_{U\times V}(g)=g\;.$$ In fact $\flat_\zW$ is a 1:1 correspondence. Indeed, any family $f_{ab}$ in $\bar\cF(\zW)$ contains $f\in\cF(\zW)$, and $f_{ab}=\zr^{U\times V}_{U_a\times V_b}(f)$, so that $\flat_\zW(f)=(f_{ab})_{ab}$. Hence, \be\label{LocIsom2}\flat_{U\times V}:\cF(U\times V)\stackrel{\sim}{\longrightarrow}\bar\cF(U\times V)\;.\ee Moreover, if $\zW'=U'\times V'\subset\zW=U\times V$, and if $f\in\cF(U\times V)$, then $$\bar\zr\,^{U\times V}_{U'\times V'}(\flat_{U\times V}(f))=\flat_{U'\times V'}(\zr^{U\times V}_{U'\times V'}(f))\;,$$ since both sides are made of the family of restrictions $\zr^{U\times V}_{U_a\times V_b}(f)$, for all $U_a\times V_b\subset U'\times V'$.\medskip

Let now $U_i\times V_j\in\frak B$ be a Cartesian product of $\Zn$-chart domains, and let $\zW\subset U_i\times V_j$ be any open subset. Recall Definition \eqref{LocSheaf1}. Since $U_a\subset U_i$ and $V_b\subset V_j$ are $\Zn$-chart domains, we get in view of \eqref{CompProdZ2n} and of \eqref{DefProdMan},
$$\cF(U_a\times V_b)=\cO_M(U_a)\widehat\0\,\cO_N(V_b)=\Ci(U_a)[[\xi_i]]\widehat\0\,\Ci(V_b)[[\zh_j]]\simeq$$
\be\label{LocSheaf2}\Ci(U_a\times V_b)[[\xi_i,\zh_j]]=\cO_{M\times N}(U_a\times V_b)\;.\ee
Due to the continuity and linearity of the restrictions $\zr^{U\times V}_{U'\times V'}$ of $\cF$ and the restrictions $\mathrm{P}^\zW_{\zW'}$ of $\cO_{M\times N}$, the restrictions in \eqref{LocSheaf1} coincide with the corresponding restrictions $\mathrm P$ of the structure sheaf of the product $\Zn$-manifold, as both reduce to the same restrictions of classical functions. It follows from \eqref{LocSheaf1} and \eqref{LocSheaf2} that any family $f_{ab}\in\bar\cF(\zW)$ is made of $\Zn$-functions $f_{ab}\in\cO_{M\times N}(U_a\times V_b)$, which are defined on the cover of $\zW$ by all the $U_a\times V_b\subset\zW$, and whose $\mathrm P$-restrictions coincide on all intersections. Hence, any family $f_{ab}\in\bar\cF(\zW)$ can be glued in the sheaf $\cO_{M\times N}$ and thus provides a unique $\Zn$-function $f\in\cO_{M\times N}(\zW)$ such that $\mathrm{P}^{\zW}_{U_a\times V_b}(f)=f_{ab}$. The resulting map $\frak{b}_\zW:\bar\cF(\zW)\to\cO_{M\times N}(\zW)$ is clearly injective. It is also surjective, since the restrictions $f_{ab}:=\mathrm{P}^{\zW}_{U_a\times V_b}(f)$ of any $f\in\cO_{M\times N}(\zW)$ define a family $f_{ab}\in\bar\cF(\zW)$ whose image by $\frak{b}_\zW$ is $f$. Therefore, if $\zW\subset U_i\times V_j$, we have a 1:1 correspondence \be\label{LocIsom3}\frak{b}_\zW:\bar\cF(\zW)\stackrel{\sim}{\longrightarrow}\cO_{M\times N}(\zW)\;.\ee Moreover, \be\label{LocSheaf3}\mathrm{P}^\zW_{\zW'}\circ\frak{b}_\zW=\frak{b}_{\zW'}\circ\bar\zr\,^\zW_{\zW'}\;.\ee Indeed, if $(f_{ab})_{ab}\in\bar\cF(\zW)$, the {\small LHS} map sends this family $f_{ab}$, $U_a\times V_b\subset\zW$, first to the unique $f\in\cO_{M\times N}(\zW)$ such that $\mathrm{P}^{\zW}_{U_a\times V_b}(f)=f_{ab}$, and then to the restriction $\mathrm{P}^\zW_{\zW'}(f)$. The {\small RHS} map sends this family first to the subfamily $f_{\za\zb}$, $U_\za\times V_\zb\subset\zW'$, then to the unique $g\in\cO_{M\times N}(\zW')$ such that $\mathrm{P}^{\zW'}_{U_\za\times V_\zb}(g)=f_{\za\zb}$. It is clear that $g=\mathrm{P}^\zW_{\zW'}(f)$. Hence, $$\frak{b}:\bar\cF|_{U_i\times V_j}\stackrel{\sim}\longrightarrow\cO_{M\times N}|_{U_i\times V_j}$$ is a presheaf isomorphism, so that $\bar\cF|_{U_i\times V_j}$ is a sheaf.\medskip

We denote the sheafification of the presheaf $\bar\cF$ by $\zvf:\bar\cF\to\bar\cF^+$ ($\bar\zr^+$ refers to the restrictions of $\bar\cF^+$). Recall that any presheaf and its sheafification have the same stalks, i.e., that the maps $$\zvf_{m,n}:\bar\cF_{m,n}\stackrel{\sim}\longrightarrow\bar\cF^+_{m,n}\;,$$ $(m,n)\in M\times N$, induced on stalks by the presheaf morphism $\zvf$ are isomorphisms. Therefore, the sheaf morphism $\zvf|_{U_i\times V_j}:\bar\cF|_{U_i\times V_j}\to\bar\cF^+|_{U_i\times V_j}$ is a sheaf isomorphism. This means that $$\zvf_{\zW}:\bar\cF(\zW)\stackrel{\sim}{\longrightarrow}\bar\cF^+(\zW)$$ is an isomorphism, or, here, a 1:1 correspondence, for any open $\zW\subset U_i\times V_j$, so that $$\iota_\zW:=\frak{b}_\zW\circ\zvf_\zW^{-1}:\bar\cF^+(\zW)\stackrel{\sim}\longrightarrow\cO_{M\times N}(\zW)\;$$ is also 1:1. In particular, for any $U_i\times V_j\in\frak{B}$, we have \be\label{GlobIsom1}\iota_{U_i\times V_j}:\bar\cF^+(U_i\times V_j)\stackrel{\sim}\longrightarrow\cO_{M\times N}(U_i\times V_j)\;.\ee Since $\zvf$ commutes with restrictions, we get, for any $U_{\frak i}\times V_{\frak j}\subset U_i\times U_j$, $$\bar\zr\,^{U_i\times V_j}_{U_{\frak i}\times V_{\frak j}}\circ\zvf^{-1}_{U_i\times U_j}=\zvf^{-1}_{U_{\frak i}\times V_{\frak j}}\circ \left.\bar\zr^+\right.^{U_i\times U_j}_{U_{\frak i}\times V_{\frak j}}\;,$$ so that, when taking also \eqref{LocSheaf3} into account, we obtain \be\label{GlobIsom2}\mathrm{P}^{U_i\times V_j}_{U_{\frak i}\times V_{\frak j}}\circ \iota_{U_i\times V_j}=\iota_{U_{\frak i}\times V_{\frak j}}\circ \left.\bar\zr^+\right.^{U_i\times U_j}_{U_{\frak i}\times V_{\frak j}}\;.\ee Due to \eqref{GlobIsom1} and \eqref{GlobIsom2}, the map $\iota$ is a $\frak B$-sheaf isomorphism between $\bar\cF^+$ and $\cO_{M\times N}$ viewed as $\frak B$-sheaves. Since a $\frak B$-sheaf morphism extends to a unique sheaf morphism, there exists a sheaf isomorphism $$\mathrm{I}:\bar\cF^+\stackrel{\sim}{\longrightarrow}\cO_{M\times N}\;.$$ 

The morphism $\mathrm{I}$ is actually an isomorphism of sheaves of $\Zn$-commutative $\R$-algebras. It suffices to show that $\imath$ is an isomorphism of $\frak B$-sheaves of such algebras, i.e., that $\imath_{U_i\times V_j}$ is a morphism of $\Zn$-graded unital $\R$-algebras. We will prove that $\imath_\zW$, $\zW\subset U_i\times V_j$, is an algebra morphism, leaving the remaining checks to the reader. The space $\cF(U\times V)$ is a nuclear Fr\'echet algebra, because it is the completed tensor product of nuclear Fr\'echet algebras \cite[Lemma 1.2.13]{Emerton}. Its multiplication $\bullet_\cF$ is continuous. It is given by $$\sum_{i=0}^\infty f_i\0 g_i\bullet_\cF\sum_{j=0}^\infty h_j\0 k_j=\sum_{i=0}^\infty\sum_{j=0}^\infty(-1)^{\langle g_i,h_j\rangle}(f_ih_j)\0 (g_ik_j)\;.$$ The multiplication $\bullet_\cF$ induces a multiplication $\bullet_{\bar\cF}$ on $\bar\cF(\zW)$, $\zW\subset M\times N$, which is defined by $$(f_{ab})_{ab}\bullet_{\bar\cF}(g_{ab})_{ab}=(f_{ab}\bullet_\cF g_{ab})_{ab}\;.$$ Addition and scalar multiplication on $\bar\cF(\zW)$ are defined similarly. As $\bar\cF$ is thus a presheaf of algebras, its sheafification $\bar\cF^+$ is a sheaf of algebras and the $\zvf_\zW:\bar\cF(\zW)\to\bar\cF^+(\zW)$ are algebra morphisms, see Subsection \ref{SheafiPShTA} of the Appendix. For $\zW\subset U_i\times V_j$, this morphism $\zvf_\zW$ is an algebra isomorphism and so is $\zvf_\zW^{-1}$. The map ${\frak b}_\zW:\bar\cF(\zW)\to \cO_{M\times N}(\zW)$ associates to each $(f_{ab})_{ab}$ the $f$ such that $\mathrm{P}^\zW_{U_a\times V_b}(f)=f_{ab}$. Hence, the image by ${\frak b}_\zW$ of a product $(f_{ab})_{ab}\bullet_{\bar\cF}(g_{ab})_{ab}$ is the function $h$ such that $\mathrm{P}^\zW_{U_a\times V_b}(h)=f_{ab}\bullet_\cF g_{ab}$. On the other hand, the product $f\cdot g$ in $\cO_{M\times N}(\zW)$ of the images by ${\frak b}_\zW$ satisfies $$\mathrm{P}^{\zW}_{U_a\times V_b}(f\cdot g)=\mathrm{P}^{\zW}_{U_a\times V_b}(f)\cdot \mathrm{P}^{\zW}_{U_a\times V_b}(g)=f_{ab}\cdot g_{ab}=f_{ab}\bullet_\cF g_{ab}\;,$$ see Theorem \ref{FundaIsom}. It follows that $h=f\cdot g$. The map ${\frak b}_\zW$ is in fact an algebra morphism. Finally $\imath_\zW={\frak b}_\zW\circ\zvf_{\zW}^{-1}$ is a morphism of algebras, as needed.
\end{proof}

\begin{rem}\label{ProblRem} In view of \eqref{LocIsom3}, \eqref{LocIsom2}, and \eqref{LocIsom1} $(\,$as well as in view of \eqref{DefProdMan} and \eqref{CompProdZ2n}$\,)$, it is clear that, for any $U_i\times V_j\in\frak{B}$, we have $$\label{FundaProbl}\cO_{M\times N}(U_i\times V_j)=\cO_M(U_i)\widehat\0\,\cO_N(V_j)\;,$$ but we were unable to convince ourselves that the same holds true for any $U\times V\in\cB$.\end{rem}

Indeed, it is well-known that tensor products of sheaves (and in particular completed tensor products of function sheaves) require a sheafification (see \cite[Section 3]{Vily}). However, section spaces of the sheafification of a presheaf do not agree with the corresponding section spaces of the presheaf.\medskip

On the other hand, attempts to get rid of the problem in Remark \ref{ProblRem} using the reconstruction results from Section \ref{ReconstructionSection} below, are not really promising.\medskip

Further, although $\cO_{M\times N}$ and $\bar\cF$ are two presheaves that extend the $\frak B$-presheaf $\cO_M\,\widehat\0\,\cO_N$, they do not necessarily coincide: $\frak B$-sheaves have unique extensions, but $\frak B$-presheaves do not. Indeed, to show that $\cO_{M\times N}\simeq\bar\cF$, we would have to decompose sections of $\cO_{M\times N}$ into sections of the $\frak B$-presheaf and then reglue them in $\bar\cF$, which is impossible, since $\bar\cF$ is only a presheaf. \medskip

There is actually a condition for the presheaf $\bar\cF$ that extends the $\cB$-presheaf $\cF=\cO_M\,\widehat\0\,\cO_N$ to be a sheaf.\medskip

The explanation of this result needs some preparation.\medskip

For any open $U\times V\subset M\times N$, we set $\cO^M_{M\times N}(U\times V):=\cO_M(U)$. Similarly, for any open $U'\times V'\subset U\times V$, we define $$r^{U\times V}_{U'\times V'}:\cO^M_{M\times N}(U\times V)\to\cO^M_{M\times N}(U'\times V')$$ to be $\zr^U_{U'}:\cO_M(U)\to\cO_M(U')$. It is straightforwardly checked that $\cO^M_{M\times N}$ is a nuclear Fr\'echet sheaf of algebras, hence, in particular a nuclear locally convex topological sheaf of algebras. The assignment $$\cF:U\times V\mapsto \cO^M_{M\times N}(U\times V)\,\widehat\0\,\cO^N_{M\times N}(U\times V)=\cO_M(U)\,\widehat\0\,\cO_N(V)$$ defines a presheaf $\bar\cF$ on $M\times N$. Applying \cite[Equation (2.2)]{Mal}, we {\it would} get \be\label{False}\cO_{M\times N}(U\times V)\simeq\bar\cF^+(U\times V)\simeq \cO_M(U)\,\widehat\0\,\cO_N(V)\;,\ee for any open $U\times V\subset M\times N$, {\it if} $\cO^M_{M\times N}$ or $\cO^N_{M\times N}$ determined a topologically dual weakly flabby precosheaf.\medskip

Just as a presheaf $\cG$ on a topological space $T$ with values in a concrete category $\tt C$ is a contravariant functor $\cG:{\tt Open}(T)^{\op{op}}\to\tt C$, a {\it precosheaf} $\cH$ on $T$ with values in $\tt C$ is a covariant functor $\cH:{\tt Open}(T)\to\tt C$. The point is that, for open subsets $V\subset U$, there is a $\tt C$-morphism $e^V_U:\cH(V)\to\cH(U)$, which we refer to as {\it extension morphism}. The relevant example for our purpose is the topologically dual precosheaf $\cV'$ of vector spaces of a topological sheaf $\cV$ of vector spaces on a Hausdorff space $T$. This precosheaf is defined, for any open $U\subset T$, by $$\cV'(U)=\h_{\tt TVS}(\cV(U),\R)\;,$$ and, for any open subsets $V\subset U$, by $$e^V_U=\,^t\!R^U_V:\cV'(V)\ni\ell\mapsto [(^t\!R^U_V)\ell:\cV(U)\ni v\mapsto \ell(R^U_Vv)\in\R]\in\cV'(U)\;,$$ where $R^U_V$ denotes the restriction in $\cV$. The precosheaf $\cV'$ is {\it weakly flabby} if, for any open $U\subset T$, the morphism $e^U_T:\cV'(U)\to\cV'(T)$ is surjective.\medskip

We should prove that the topologically dual precosheaf of $\cV=\cO^M_{M\times N}$ is weakly flabby, i.e., that, for any open $U\subset M$ and for any $L\in\h_{\tt TVS}(\cO_M(M),\R)$, there exists $\ell\in\h_{\tt TVS}(\cO_M(U),\R)$, such that $^t\!\zr^M_U\ell = L$. It turns out that this condition is not satisfied, so that we cannot conclude that \eqref{False} holds. Indeed, assume that the condition is satisfied, so that it is in particular valid for $\cM=(\R,\Ci_\R)$. Choose now any $x\in\R$ and any open interval $I\subset\R$ that does not contain $x$. The evaluation map $$\ze_x:\Ci_\R(\R)\ni f\mapsto f(x)\in\R$$ is linear. It is also continuous, since there exists a compact $C\subset\R$ that contains $x$ and since, for any $f\in\Ci_\R(\R)$, we have $|f(x)|\le\sup_C|f|$. In view of our assumption, there exists $\ell\in\h_{\tt TVS}(\Ci_\R(I),\R),$ such that, for any $f\in\Ci_\R(\R)$, we have $\ell(f|_I)=f(x)$. If we take now two functions $f,g\in\Ci_\R(\R)$ that coincide in $I$ and have different values at $x$, we get the contradiction $$f(x)=\ell(f|_I)=\ell(g|_I)=g(x)\;.$$

\subsection{Categorical products of $\Zn$-manifolds}

We recommend to first read Subsections \ref{SheafiPShTA} and \ref{SheafiDirIm} of the Appendix.

\begin{lem}
Let $\cM_1,\cM_2\in\Zn\tt Man$. The presheaf $\bar\cF$ considered in Theorem \ref{IsomExtShfifi} is an object of the category ${\tt PSh}(M_1\times M_2,\tt LCTAlg)$ of presheaves of locally convex topological algebras over $M_1\times M_2$.
\end{lem}

\begin{proof}
In the proof of Theorem \ref{IsomExtShfifi} we showed that $\bar\cF\in{\tt PSh}(M_1\times M_2,\tt Alg)$ for the obvious restrictions and algebra operations.\medskip

Recall that, for any open $\zW\subset M_1\times M_2$, the algebra $\bar\cF(\zW)$ is given by $$\bar\cF(\zW):=\{(f_{ab})_{ab}:f_{ab}\in\cF(U_a\times V_b), U_a\times V_b\subset\zW\;,\text{such that}\;\zr^{U_a\times V_b}_{U_{a\frak a}\times V_{b\frak b}}(f_{ab})=\zr^{U_{\frak a}\times V_{\frak b}}_{U_{a\frak a}\times V_{b\frak b}}(f_{\frak{ab}})\}\;$$ and is thus a subalgebra of $$\prod_{U_a\times V_b\subset\,\zW}\cF(U_a\times V_b)=\prod_{U_a\times V_b\subset\,\zW}\cO_{M_1}(U_a)\,\widehat\0\;\cO_{M_2}(V_b)\;.$$ We equip $\bar\cF(\zW)$ with the topology induced by the product of the topologies of the locally convex topological algebras ({\small LCTA}-s) $\cO_{M_1}(U_a)\,\widehat\0\;\cO_{M_2}(V_b)$. Since a product of {\small LCTVS}-s and a subspace of a {\small LCTVS} are themselves {\small LCTVS}-s, the algebra $\bar\cF(\zW)$ is a {\small LCTVS}. Its multiplication $$(f_{ab})_{ab}\bullet_{\bar\cF}(g_{ab})_{ab}=(f_{ab}\bullet_\cF g_{ab})_{ab}$$ is continuous, since the multiplication $\bullet_\cF$ is, see Proof of Theorem \ref{IsomExtShfifi} and Lemma \ref{PreRes1}. Hence, the space $\bar\cF(\zW)$ is a {\small LCTA}.\medskip

A restriction $\bar\zr^{\zW}_{\zW'}$ ($\zW'\subset\zW$) -- it sends any family $(f_{ab})_{ab}$ of $\bar\cF(\zW)$ indexed by the $U_a\times V_b\subset\zW$ to the family $(f_{ab})_{ab}$ of $\bar\cF(\zW')$ indexed by the $U_a\times V_b\subset\zW'$ -- is known to be an algebra morphism. It is continuous, since it is continuous as a map $$\bar\zr^{\zW}_{\zW'}:\prod_{U_a\times V_b\subset\,\zW}\cF(U_a\times V_b)\to \prod_{U_a\times V_b\subset\,\zW'}\cF(U_a\times V_b)\;,$$ in view of the definition of the product topology.\end{proof}

\begin{theo}\label{FinProds} The category $\tt \Zn Man$ has all finite products.\end{theo}

\begin{proof}
Since $\tt \Zn Man$ has a terminal object (see Corollary \ref{IniTer}), it suffices to prove that it has binary products. Let $\cM_1,\cM_2\in\tt\Zn Man$. We will show that the product $\Zn$-manifold $\cM_1\times\cM_2$ (see Definition \ref{ProdZnMan}) is the categorical binary product of $\cM_1$ and $\cM_2$.\medskip

We first define $\Zn$-morphisms $$\zP_i=(\zp_i,\zp_i^*):\cM_1\times\cM_2\to\cM_i,\quad i\in\{1,2\}\;.$$ The base maps $\zp_i:M_1\times M_2\to M_i$ are the canonical smooth projections. In the following, we consider the case $i=1$ and use the notation introduced above.\medskip

The maps $$T_U:\cO_{M_1}(U)\ni f\mapsto f\0 1\in (\zp_{1,*}\bar\cF)(U)\;,$$ $U\in{\tt Open}(M_1)$, define a morphism $T:\cO_{M_1}\to \zp_{1,*}\bar\cF$ in ${\tt PSh}(M_1,\tt LCTAlg)$. Indeed, we have $\cO_{M_1},\,\zp_{1,*}\bar\cF\in{\tt PSh}(M_1,{\tt LCTAlg})$, see \cite[Theorem 14]{BP1} and Proof of Proposition \ref{DirImShfifi}. It is easy to see that the maps $T_U$ commute with restrictions and are algebra morphisms. To show that $T_U$ is continuous, it suffices to check that the linear map \be\label{CompA}T_U:\cO_{M_1}(U)\ni f\mapsto f\0 1\in \cO_{M_1}(U)\0\cO_{M_2}(M_2)\ee is continuous for the projective tensor topology on the target. We apply Theorem \ref{ProjTopSN} and Proposition \ref{CoSN}. Let $p_{C,D}\0 p_{K,\zD}$ be any of the seminorms that induce the projective tensor topology. Recall that $C\subset U$ and $K\subset M_2$ are compact subsets, and that $D$ and $\zD$ are differential operators acting on $\cO_{M_1}(U)$ and $\cO_{M_2}(M_2)$, respectively. We must prove that there is a finite number of seminorms $p_{C_k,D_k}$ on the source and a constant $C>0$, such that $$p_{C,D}(f)\cdot p_{K,\zD}(1)\le C\max_kp_{C_k,D_k}(f)\;,$$ for any $f\in\cO_{M_1}(U)$. It suffices to use a single source seminorm $p_{C_k,D_k}$, namely $p_{C,D}$. Indeed, if ${\frak C}=p_{K,\zD}(1)=0$, the previous condition is satisfied with $C=1$, and if ${\frak C}>0$, it is satisfied with $C=\frak C$.\medskip

It follows from Proposition \ref{LCTAShfifi} that $T^+:\cO_{M_1}^+\to (\zp_{1,*}\bar\cF)^+$ is a morphism in ${\tt Sh}(M_1,\tt LCTAlg)$. Moreover, in view of Proposition \ref{DirImShfifi}, there is an ${\tt Sh}(M_1,{\tt LCTAlg})$-morphism $\iota: (\zp_{1,*}\bar\cF)^+\to \zp_{1,*}\,\bar\cF^+$. The composite $\iota\circ T^+:\cO_{M_1}^+\to \zp_{1,*}\bar\cF^+$ is a morphism in ${\tt Sh}(M_1,\tt LCTAlg)$. Proposition \ref{LCTAShfifi} and Theorem \ref{IsomExtShfifi} allow us to interpret it as a morphism \be\label{CompB}\zp_{1}^*=\iota\circ T^+:\cO_{M_1}\to\zp_{1,*}\,\cO_{M_1\times M_2}\ee in ${\tt Sh}(M_1,\tt Alg)$. The map $\zp_1^*$ is actually a morphism of sheaves of $\Zn$-commutative associative unital algebras, so that the map $\zP_1=(\zp_1,\zp_1^*):\cM_1\times\cM_2\to\cM_1$ is a morphism in $\Zn\tt Man$ (the results of the appendix we use in this proof extend obviously to the graded unital setting: when speaking in the rest of the proof about algebras, we actually mean $\Zn$-commutative associative unital algebras).\medskip

It remains to check the universality of our construction. Let ${\cN}\in\Zn\tt Man$ and let $\Phi_i=(\zf_i,\zf_i^*):{\cN}\to{\cM}_{i}$ be morphisms in $\Zn{\tt Man}$. We will prove that there exists a unique $\Zn{\tt Man}$-morphism $\Psi=(\psi,\psi^*):{\cN}\to\cM_1\times\cM_2$, such that $\zP_i\circ \Psi=\zF_i$, for both $i$.\medskip

We set $\psi=(\zf_1,\zf_2):N\to M_1\times M_2\,$. As for $\psi^*$, observe that, for any open $U_i\subset M_i$, the map $\zf^*_{i,U_i}:\cO_{M_i}(U_i)\to\cO_N(\zf_i^{-1}(U_i))$ is a continuous algebra ($\Zn$-graded unital algebra) morphism, see \cite[Theorem 19]{BP1}. Denote now by $V=V_1\cap V_2$ the open subset $$V=V_1\cap V_2=\zf_1^{-1}(U_1\,)\cap\zf_2^{-1}(U_2)=\psi^{-1}(U_1\times U_2)\subset N\;,$$ and denote by $m_V=-\cdot-$ the multiplication of $\cO_N(V)$. In view of \cite[Theorem 14]{BP1} and Propositions \ref{ComplTensCoAlgMor} and \ref{ExtMult}, the map $${\frak p}_{U_1\times U_2}=\widehat m_V\circ(\zr^{V_1}_V\widehat\0\,\zr^{V_2}_{V})\circ(\zf^*_{1,U_1}\widehat\0\,\zf^*_{2,U_2}):\cO_{M_1}(U_1)\widehat\0\,\cO_{M_2}(U_2)\to \cO_N(V)\;,$$ where $\zr^{V_1}_V$ and $\zr^{V_2}_V$ are restrictions in $\cO_N$, is a continuous algebra morphism between nuclear Fr\'echet algebras. The maps \be\label{CompC}{\frak p}_{U_1\times U_2}:\cF(U_1\times U_2)\ni\sum_{j=0}^\infty f_j\0 g_j\mapsto \sum_{j=0}^\infty\,\zr^{V_1}_V\zf^*_{1,U_1}\,f_j\,\cdot\,\zr^{V_2}_V\zf^*_{2,U_2}\,g_j\in (\psi_*\cO_N)(U_1\times U_2)\ee define a morphism of $\cB$-presheaves of locally convex topological algebras. The latter extends to a morphism $\bar{\frak p}:\bar\cF\to\psi_*\cO_N$ of presheaves of locally convex topological algebras over $M_1\times M_2\,$.\medskip

In view of Propositions \ref{LCTAShfifi} and \ref{DirImShfifi}, there are morphisms $$\bar{\frak p}^+:\bar\cF^+\to (\psi_*\cO_N)^+\quad\text{and}\quad\iota: (\psi_*\cO_N)^+\to\psi_*\cO_N^+$$ in ${\tt Sh}(M_1\times M_2,{\tt LCTAlg})$. As above, we can view their composite as a morphism \be\label{CompD}\psi^*=\iota\circ\bar{\frak p}^+:\cO_{M_1\times M_2}\to \psi_*\cO_N\ee of sheaves of algebras, so that $\Psi=(\psi,\psi^*):{\cN}\to{\cM}_1\times{\cM}_2$ is a morphism in $\Zn\tt Man$.\medskip

To prove that $\zP_i\circ\Psi=\zF_i$, it suffices to show that $$\psi^*_{M_1\times M_2}\circ\zp^*_{i,M_i}=\zb(\zP_i\circ\Psi)=\zb(\zF_i)=\zf_{i,M_i}^*\;,$$ see Theorem \ref{Thm:SpaceAlgebraMorphBijection}. The latter is a straightforward consequence of Equations \eqref{CompA}, \eqref{CompB}, \eqref{CompC}, \eqref{CompD}, \eqref{CompE}, and \eqref{CompF}.\medskip

Let now $X=(\chi,\chi^*):{\cN}\to\cM_1\times\cM_2$ be another $\Zn\tt Man$-morphism that satisfies $\zP_i\circ X=\zF_i$. As the category of smooth manifolds has finite products, we get $\chi=\psi$. We will check that $\zb(X)=\zb(\Psi)$, i.e., that, for all $\zs\in\bar\cF^+(M_1\times M_2)$, we have $$\chi^*_{M_1\times M_2}\zs=\psi^*_{M_1\times M_2}\zs\in\cO_N(N)\;.$$ It suffices to show that these sections coincide in a neighborhood of an arbitrary point $n_0\in N$. We use the compact notation $m\in M$ instead of $(m_1,m_2)\in M_1\times M_2$. Recall that $$\zs=([s]_m)_{m\in M}\;,$$ where $s\in\cF(U)$ ($U=U_1\times U_2\ni m$) reads $s=\sum_{j=0}^\infty f_j\0 g_j\,$ ($f_j\in\cO_{M_1}(U_1)$, $g_j\in\cO_{M_2}(U_2)$).\medskip

In view of \eqref{CompC}, \eqref{CompD}, \eqref{CompE}, and \eqref{CompF}, \be\label{PullBacksUniv}\psi^*_{M_1\times M_2}\zs=([\sum_{j=0}^\infty\,\zr^{V_1}_V\zf^*_{1,U_1}\,f_j\,\cdot\,\zr^{V_2}_V\zf^*_{2,U_2}\,g_j]_n)_{n\in N}\;,\ee where we take the germ at $n$ of the section induced by the representative $s$ of the germ at $\psi(n)\in M$. Let $m_0=\psi(n_0)\in M$. Since $s$ is constant in a neighborhood $U_0=U_{1,0}\times U_{2,0}$ of $m_0$, the representative in the {\small RHS} of the preceding equation is constant in the neighborhood $$V_0=V_{1,0}\cap V_{2,0}=\zf_1^{-1}(U_{1,0})\cap\zf_2^{-1}(U_{2,0})=\psi^{-1}(U_0)\;$$ of $n_0\,$. Hence, $$\zr^{N}_{V_0}\,\psi^*_{M_1\times M_2}\zs = \sum_{j=0}^\infty\,\zr^{V_{1,0}}_{V_0}\zf^*_{1,U_{1,0}}\,f_j\,\cdot\,\zr^{V_{2,0}}_{V_0}\zf^*_{2,U_{2,0}}\,g_j\in\cO_N(V_0)\;.$$

On the other hand, since $\zP_i\circ X=\zF_i$, we have, for any open $U_1\subset M_1$ and any $f\in\cO_{M_1}(U_1)$, $$\chi^*_{U_1\times M_2}([f\0 1]_m)_{m\in U_1\times M_2}=\zf^*_{1,U_1}f\;,$$ due to \eqref{CompA}, \eqref{CompB}, \eqref{CompE}, and \eqref{CompF}. For any open $U_2\subset M_2$, we get \be\label{Info}\chi^*_{U}([f\0 1]_m)_{m\in U}=\zr_{\psi^{-1}(U)}^{\psi^{-1}(U_1\times M_2)}\chi^*_{U_1\times M_2}([f\0 1]_m)_{m\in U_1\times M_2}=\zr^{V_1}_V\zf^*_{1,U_1}f\;.\ee An analogous result holds for $i=2$. Observe now that $$\zs=([\sum_{j=0}^\infty f_j\0 g_j]_m)_{m\in M}=(\zp^m_{U}\lim_{N}\sum_{j=0}^N f_j\0 g_j)_{m\in M}=$$ \be\label{SigmaSheafifi}(\lim_{N}\sum_{j=0}^N [f_j\0 g_j]_m)_{m\in M}=\sum_{j=0}^\infty([f_j\0 g_j]_m)_{m\in M}\;,\ee see Proof of Proposition \ref{LCTAShfifi}. As the pullbacks $\chi^*$ and the restriction $\zr^N_{V_0}$ are continuous algebra morphisms, we get $$\zr^N_{V_0}\,\chi^*_{M_1\times M_2}\zs=\sum_{j=0}^\infty\zr^{\psi^{-1}(M)}_{\psi^{-1}(U_0)}\,\chi^*_{M}([f_j\0 g_j]_m)_{m\in M}=$$ $$\sum_{j=0}^\infty\,\chi^*_{U_0}([f_j\0 1]_m)_{m\in U_0}\cdot \chi^*_{U_0}([1\0 g_j]_m)_{m\in U_0}=\sum_{j=0}^\infty\zr^{V_{1,0}}_{V_0}\,\zf_{1,U_{1,0}}^*f_j\cdot \zr^{V_{2,0}}_{V_0}\,\zf_{2,U_{2,0}}^*g_j\;,$$ due to \eqref{Info}.
\end{proof}

\subsection{Products of $\Zn$-morphisms}

We use again abbreviations of the type $n=(n_1,n_2)\in N=N_1\times N_2$ (which we introduced in the proof of Theorem \ref{FinProds}).

\begin{prop}\label{PropProdZ2nMor}
Let $\Psi_i:{\cM}_i\to{\cN}_i$, $i\in\{1,2\}$, be a $\Zn$-morphism with base map $\psi_i$ and pullback sheaf morphism $\psi_i^*$. Due to the universality of the product of $\Zn$-manifolds, there is a canonical $\Zn$-morphism $$\Psi=\Psi_1\times\Psi_2:{\cM_1\times\cM_2}\to{\cN}_1\times{\cN}_2\;.$$ Its base map is $\psi=\psi_1\times\psi_2$ and its pullback sheaf morphism $\psi^*=(\psi_1\times\psi_2)^*$ is given, for each open subset $\zW\subset N_1\times N_2,$ by $$\psi_\zW^*=\left(\prod_{m\in\psi^{-1}(\zW)}\, (\psi_1^*\widehat\0\,\psi_2^*)_{\psi(m)}\right)\circ (p_{\psi(m)})_{m\in\psi^{-1}(\zW)}\;.$$ The first map in the {\small RHS} is the product of the morphisms between stalks induced by the morphism $$\psi_1^*\widehat \0\,\psi_2^*:\bar\cF_N\to \psi_*\bar\cF_M$$ (of presheaves of locally convex topological algebras), where $\bar\cF_{N}$ is the presheaf defined by $\cO_{N_1}\widehat\0\,\cO_{N_2}$, and similarly for $\bar\cF_M$. The second map in the {\small RHS} is the tuple of morphisms $$p_{\psi(m)}:\prod_{n\in\zW}\bar\cF_{N,n}\to \bar\cF_{N,\psi(m)}\;.$$
\end{prop}

To understand this claim, recall that $$\psi_\zW^*:\cO_N(\zW)\simeq\bar\cF_N^+(\zW)\subset\prod_{n\in\zW}\bar\cF_{N,n}\to \cO_{M}(\psi^{-1}(\zW))\simeq\bar\cF_M^+(\psi^{-1}(\zW))\subset \prod_{m\in\psi^{-1}(\zW)}\bar\cF_{M,m}\;.$$ Note now that $$(p_{\psi(m)})_{m\in\psi^{-1}(\zW)}:\cO_N(\zW)\to\prod_{m\in\psi^{-1}(\zW)}\bar\cF_{N,\psi(m)}$$ and that $$\prod_{m\in\psi^{-1}(\zW)}\, (\psi_1^*\widehat\0\,\psi_2^*)_{\psi(m)}:\prod_{m\in\psi^{-1}(\zW)}\bar\cF_{N,\psi(m)}\to \prod_{m\in\psi^{-1}(\zW)}\bar\cF_{M,m}\;,$$ so that the composite of these maps may coincide with $\psi^*_\zW$.

\begin{proof}
For $i\in\{1,2\}$, we denote by $\zP_i^S$ (resp., $\zP_i^T$) the $\Zn$-morphism $\zP_i^S:\cM_1\times\cM_2\to\cM_i$ (resp., $\zP_i^T:\cN_1\times\cN_2\to\cN_i$). The composite $\zF_i=\Psi_i\circ\zP_i^S$ is a $\Zn$-morphism $\zF_i:\cM_1\times\cM_2\to{\cN}_i$. In view of the universality of the product $\cN_1\times\cN_2$, there is a unique $\Zn$-morphism $\Psi:\cM_1\times\cM_2\to\cN_1\times\cN_2$, such that $\zP_i^T\circ\Psi=\zF_i$. We denote this morphism $\Psi$ by $\Psi_1\times\Psi_2$ and refer to it as the product of the $\Zn$-morphisms $\Psi_i$. We showed in the proof of Theorem \ref{FinProds} that the base map of $\Psi$ is $$\psi=(\psi_1\circ\zp^S_1,\psi_2\circ\zp^S_2)=\psi_1\times\psi_2\;.$$ We investigate now the pullback morphisms $\psi^*$. Let $\zW\subset N$ be open, so that $\psi^*_\zW:\cO_{N}(\zW)\to \cO_{M}(\psi^{-1}(\zW))$, and let $\zs\in $ $\cO_{N}(\zW)\simeq\bar\cF_N^+(\zW)$. Recall once again that $\zs=([s]_n)_{n\in\zW}$, where $s\in\cF_N(U)$ ($U=U_1\times U_2\ni n$) reads $$s=\sum_{j=0}^\infty f_j\0 g_j\quad (f_j\in\cO_{N_1}(U_1), g_j\in\cO_{N_2}(U_2))\;.$$ It follows from Equation \eqref{PullBacksUniv} that \be\label{IntermedPsisigma}\psi^*_\zW\,\zs=([\sum_{j=0}^\infty\,\zr^{V_1}_V(\zp^{S}_1)^*_{W_1}\psi^*_{1,U_1}\,f_j\,\cdot\,\zr^{V_2}_V(\zp^{S}_2)^*_{W_2}\psi^*_{2,U_2}\,g_j]_m)_{m\in \psi^{-1}(\zW)}\;,\ee where $$V=V_1\cap V_2,\quad V_i=(\zp_i^S)^{-1}(W_i),\quad\text{and}\quad W_i=\psi_i^{-1}(U_i)\quad\quad (\text{obviously}\;V=W_1\times W_2)\;.$$ We interpret $\psi^*_{1,U_1}\,f_j\in\cO_{M_1}(W_1)$ as $$([\psi^*_{1,U_1}\,f_j]_{m_1})_{m_1\in W_1}\in\cO^+_{M_1}(W_1)\;,$$ so that $$\zr^{V_1}_V(\zp^{S}_1)^*_{W_1}\psi^*_{1,U_1}\,f_j=([\psi^*_{1,U_1}\,f_j\0 1]_{\zm})_{\zm\in V}\;,$$ due to Equations \eqref{CompA} and \eqref{CompB}, as well as to the observation that $m_1\in W_1$ is equivalent to $m_1=\zp^S_1(m)\;(m\in V_1).$ A similar result holds for the second factor in the {\small RHS} of \eqref{IntermedPsisigma}, so that, in view of \eqref{SigmaSheafifi}, we obtain $$\psi^*_\zW\,\zs=\left(\left[\sum_{j=0}^\infty\,([\psi^*_{1,U_1}\,f_j\0 \psi^*_{2,U_2}\,g_j]_{\zm})_{\zm\in V}\right]_m\right)_{m\in \psi^{-1}(\zW)}=$$ $$\left(\left[\,([\psi^*_{1,U_1}\widehat\0\,\psi^*_{2,U_2}\sum_{j=0}^\infty\,f_j\0\,g_j]_{\zm})_{\zm\in V}\right]_m\right)_{m\in \psi^{-1}(\zW)}\;.$$ This image is a family (indexed by $m$) of elements of $\bar\cF_{M,m}^+\simeq\bar\cF_{M,m}$. The isomorphism between a stalk of a presheaf and the corresponding stalk of its sheafification is described in the proof of Lemma 6.17.2 of the Stacks Project. The description shows that $$\psi^*_\zW\,\zs=\left(\left[\,\psi^*_{1,U_1}\widehat\0\,\psi^*_{2,U_2}\sum_{j=0}^\infty\,f_j\0\,g_j\right]_m\right)_{m\in \psi^{-1}(\zW)}=$$ $$\left(\prod_{m\in\psi^{-1}(\zW)}\, (\psi_1^*\widehat\0\,\psi_2^*)_{\psi(m)}\right) \left((p_{\psi(m)})_{m\in\psi^{-1}(\zW)}\;\zs\right)\;.$$\end{proof}

\section{Appendix}

We prove and recall results on topological vector spaces and on topological algebras.

\subsection{Topological vector subspaces}

\begin{defi} A \emph{topological vector subspace} $S$ of a {\small TVS} $V$ ({\small TVSS} for short) is a subset $S\subset V$ which is a {\small TVS} for the linear operations and the topology of $V$.\end{defi}

\begin{prop} A subset $S\subset V$ of a {\small TVS} $V$ is a {\small TVSS} of $V$ if and only if $S$ is a linear subspace of $V$ and is endowed with the {\sl topology induced} by the topology of $V$.\end{prop}

\begin{proof} The restrictions to $S$ of the continuous addition and scalar multiplication in $V$ are continuous in the topology induced on $S$ by $V$.\end{proof}

Note also that, if $V$ is a {\small TVS}, and if $S\subset V$ is a {\small TVSS}, then $S$ is a {\small TVS} and the inclusion $\imath:S\ni s\mapsto s\in V$ is an injective continuous linear map. Conversely, if $S\subset V$ is a {\small TVS} and the inclusion $\imath$ is an injective continuous linear map, then the linear structure on $S\subset V$ is the same as in $V$, but $S$ can have a topology that is finer than the induced one, so that $S$ is not necessarily a {\small TVSS}.

\begin{prop}\label{TVSS} Let $\imath:V\to W$ be an injective linear map between {\small TVS}-s. When equipped with the induced topology, the linear subspace $\imath(V)$ is a {\small TVSS} of $W$. The bijective linear map $\tilde \imath:V\to \imath(V)$ is a {\small TVS}-isomorphism, i.e., a linear homeomorphism, if and only if the topology of $V$ is the {\sl initial topology} of $\imath:V\to W$. In this case, the space $V$ can be viewed as a {\small TVSS} $V\simeq \imath(V)$ of $W$.\end{prop}

\begin{proof} It suffices to prove the second claim. If $\tilde\imath$ is an isomorphism, then the topology $\cT(V)$ of $V$ is given by $$\cT(V)=\{\tilde\imath^{-1}(\imath(V)\cap U_W)=\imath^{-1}(\imath(V)\cap U_W)=\imath^{-1}(\imath(V))\cap \imath^{-1}(U_W)=\imath^{-1}(U_W):U_W\in\cT(W)\}\;,$$ hence, it is the initial topology of $\imath:V\to W$. Conversely, if $\cT(V)$ is the initial topology of $\imath$, the maps $\tilde\imath$ and $\tilde\imath^{-1}$ are continuous. Indeed, as just explained, we have $\tilde\imath^{-1}(\imath(V)\cap U_W)=\imath^{-1}(U_W)\in\cT(V)$, so that $\tilde\imath$ is continuous (but it would still be continuous for a finer topology on $V$). For $\tilde\imath^{-1}$ we have $\tilde\imath(\imath^{-1}(U_W))=\imath(V)\cap U_W\in\cT(\imath(V))$. \end{proof}

\begin{prop}\label{InjLinIniTop} Let $\imath:V\to W$ be an injective linear map from a vector space $V$ to a {\small TVS} $W$. When equipped with the initial topology of $\imath$, the linear space $V$ is a {\small TVS} and can be viewed as a {\small TVSS} of $W$.\end{prop}

\begin{proof} The linear operations on $V$ are continuous when $V$ has the initial topology. Indeed, denote by $+_V$ (resp., $+_W$) the addition in $V$ (resp., $W$), and let $\imath^{-1}(U_W)$, $U_W\in\cT(W)$, be an arbitrary open subset in $\cT(V)$. Then, $$(+_V)^{-1}(\imath^{-1}(U_W))=\{(v,v')\in V\times V:\imath(v)+_W\;\imath(v')\in U_W\}=$$ $$\{(v,v')\in V\times V:(\imath(v),\imath(v'))\in (+_W)^{-1}(U_W)\}=(\imath\times\imath)^{-1}((+_W)^{-1}(U_W))\;$$ is open in $V\times V$, since $\imath$ and $+_W$ are continuous. The case of the scalar multiplication is similar. The second claim follows now from Proposition \ref{TVSS}.\end{proof}

\begin{rem} When passing above from topological vector subspace structures on included subsets to topological vector subspace structures on injected subsets, we replaced the induced topology by the initial topology. Of course, if the injection is the inclusion, the initial topology with respect to it coincides with the induced topology.\end{rem}

\subsection{Completions of topological vector spaces}

We recall now well-known properties of the completion of a {\small TVS} \cite{Nagy}.\medskip

For any {\small TVS} $V$, there is a complete {\small TVS} $\widehat V$, and an injective linear map $\imath:V\to \widehat V$, such that the linear subspace $\imath(V)\subset\widehat V$ is dense in $\widehat V$. Moreover, when endowed with the induced topology, the image $\imath(V)\subset \widehat V$ becomes a {\small TVSS} of $\widehat V$, and the map $\tilde\imath:V\to\imath(V)$ is promoted to a linear homeomorphism, or, equivalently, to a {\small TVS}-isomorphism. It follows from Proposition \ref{TVSS} that the topology of $V$ is the initial topology of $\imath:V\to \widehat V$, and that $V\simeq\imath(V)$ is a {\small TVSS} of $\widehat V$. In short, the complete {\small TVS} $\widehat V$, which we refer to as the {\it completion} of the {\small TVS} $V$, contains $V$ as a dense {\small TVSS}:

\begin{prop}\label{Completion} The completion $\widehat V$ of $V$ contains $V$ as a dense subset, and it induces on $V$ the original topology and original linear structure.\end{prop}

\begin{rem}\label{CompletionComplete} If $V$ is already a complete {\small TVS}, then $\imath:V\to\widehat V$ is a {\small TVS}-isomorphism.\end{rem}

Further, let $S\subset V$ be a {\small TVSS} of $V$ and let $\imath:S\to V$ be the injective continuous linear inclusion. The continuous extension of $\imath$, which we denote by $\hat \imath:\widehat S\to \widehat V$, is an injective continuous linear map. If $\hat\imath(\widehat S)$ carries the induced topology, the map $\tilde{\hat{\imath}}:\widehat S\to \hat\imath(\widehat S)$ is a {\small TVS}-isomorphism. In view of Proposition \ref{TVSS} this means that the topology of $\widehat S$ is the initial topology of $\hat\imath:\widehat S\to\widehat V$, and that $\widehat S\simeq \hat \imath(\widehat S)$ is a {\small TVSS} of $\widehat V$. In short:

\begin{prop}\label{CompletionTVSS} The completion of a {\small TVSS} $S\subset V$ is a {\small TVSS} $\widehat S\subset\widehat V$ of the completion.\end{prop}

\subsection{Locally convex spaces}

\begin{prop}\label{IniTopLC} The initial topology of a linear map $\ell:V\to W$ from a vector space $V$ to a {\small LCTVS} $W$ endows $V$ with a {\small LCTVS} structure. In particular, the induced topology on a vector subspace $S\subset W$ of a {\small LCTVS} endows $S$ with a {\small LCTVS} structure.\end{prop}

\begin{proof} The topology of $W$ has a convex basis $\frak B$. As the preimage of a convex set by a linear map is convex, the family $\ell^{-1}(\frak B)$ is a convex basis of the initial topology. In view of the proof of Proposition \ref{InjLinIniTop}, the initial topology endows $V$ with a {\small LCTVS} structure. The second claim is a special case of the first one.\end{proof}

We close this subsection recalling two results.

\begin{theo}\label{ProjTopSN}
Let $V$ and $W$ be two {\small LCTVS}-s and let $(p_i)_{i\in I}$ and $(q_j)_{j\in J}$ be two families of seminorms that induce the topologies of $V$ and $W$, respectively. The projective tensor topology $\zp$ on $V\0 W$ is induced by $(p_i\0 q_j)_{ij}$. Moreover, for any $t\in V\0 W$, we have $$(p_i\0 q_j)(t)=\inf\{\sum_{k=1}^Np_i(v_k)q_j(w_k):t=\sum_{k=1}^Nv_k\0 w_k, v_k\in V, w_k\in W, N\in\N\}\;,$$ and, for any $v\in V$ and $w\in W$, we have $$(p_i\0 q_j)(v\0 w)=p_i(v)q_j(w)\;.$$
\end{theo}

\begin{prop}\label{CoSN}
Let $V$ and $W$ be two {\small LCTVS}-s and let $(p_i)_{i\in I}$ and $(q_j)_{j\in J}$ be two families of seminorms that induce the topologies of $V$ and $W$, respectively. A linear map $\ell:V\to W$ is continuous if and only if, for any $j\in J$, there exist $i_1,\ldots, i_N\in I$ and $C>0$, such that, for all $v\in V$, one has $$q_j(\ell(v))\le C\max_{k}p_{i_k}(v)\;.$$
\end{prop}

\subsection{Presheaves of topological algebras and sheafification}\label{SheafiPShTA}

In the following, we need two lemmas. \medskip

Unless otherwise stated, all cartesian products $\prod_\za T_\za$ of topological spaces $T_\za$ are endowed with the product topology, i.e., the weakest topology for which all projections $p_\zb:\prod_\za T_\za\to T_\zb$ are continuous.

\begin{lem}\label{PreRes1}
Let $(T^i_\za)_\za$ $(\,$$i\in\{1,2,3\}$$\,)$ be families of topological spaces, and let $(m_\za)_\za$ be a family $m_\za:T^1_\za\times T^2_\za\to T^3_\za$ of continuous maps. Then the map $$m:\prod_\za T^1_\za\;\times\;\prod_\za T^2_\za\ni ((\zt_\za)_\za,({\frak t}_\za)_\za)\mapsto (m_{\za}(\zt_\za,{\frak t}_\za))_\za\in\prod_\za T^3_\za$$ is continuous.\end{lem}

\begin{proof}
Since the cartesian product $\prod_\za T_\za$ of topological spaces $T_\za$ equipped with the product topology, is the product in the category of topological spaces, it follows from the universal property that a map $f:T\to\prod_\za T_\za$ from a topological space to a product space is continuous if and only if all the $p_\zb\circ f:T\to T_\zb$ are continuous. However, the map $p_\zb^3\circ m:((\zt_\za)_\za,({\frak t}_\za)_\za)\mapsto m_{\zb}(\zt_\zb,{\frak t}_\zb)$ is the composite $m_\zb\circ(p^1_\zb\times p^2_\zb)$, which is of course continuous.
\end{proof}

\begin{lem}\label{PreRes2} Let $V$ be a vector space, let $(V_\za)_\za$ be a family of {\small LCTVS}, let $\ell_\za: V_\za\to V$ be a family of linear maps, and let $V$ be equipped with the finest locally convex vector space topology, for which all the $\ell_\za$ are continuous $(\,$$V$ is then a {\small LCTVS}$\,)$. If $W$ is a {\small LCTVS}, a linear map $\ell:V\to W$ is continuous if and only if all the maps $\ell\circ\ell_\za:V_\za\to W$ are continuous.\end{lem}

\begin{proof}
See \cite[Proposition 2.3.5]{IndLim}
\end{proof}

In the present text:

\begin{defi}
A \emph{topological algebra} $(\,${\small TA} for short$\,)$ $(\,$resp., a \emph{locally convex topological algebra} $(\,${\small LCTA} for short$\,)$$\,)$ is a $(\,$real$\,)$ topological vector space $(\,$resp., a locally convex $(\,$real$\,)$ topological vector space$\,)$, with an associative, bilinear, and $(\,$jointly$\,)$ continuous multiplication. A \emph{morphism of topological algebras} $(\,$resp., a \emph{morphism of locally convex topological algebras}$\,)$ is a continuous algebra morphism. We denote the category of topological algebras $(\,$resp., of locally convex topological algebras$\,)$ and morphisms between them by $\tt TAlg$ $(\,$resp., by $\tt LCTAlg$$\,)$.
\end{defi}

The category of Fr\'echet algebras is a full subcategory of the category of locally convex topological algebras, itself a full subcategory of the category of topological algebras.\medskip

We denote by ${\tt PSh}(T,{\tt TAlg})$ (resp., ${\tt Sh}(T,{\tt TAlg})$) the category of presheaves (resp., sheaves) of {\small TA}-s over a topological space $T$. Similarly, the category ${\tt PSh}(T,{\tt LCTAlg})$ (resp., ${\tt Sh}(T,$ ${\tt LCTAlg})$) is the category of presheaves (resp., sheaves) of {\small LCTA}-s over $T$, and the category ${\tt PSh}(T,{\tt Alg})$ (resp., ${\tt Sh}(T,{\tt Alg})$) is the category of presheaves (resp., sheaves) of algebras over $T$.

\begin{prop}\label{LCTAShfifi}
Denote by $+$ the sheafification functor $$+:{\tt PSh}(T,{\tt Alg})\to{\tt Sh}(T,{\tt Alg}):\op{For}\;,$$ i.e., the left adjoint of the forgetful functor $\op{For}$. If $\cF\in{\tt PSh}(T,{\tt TAlg})$, we have $\cF^+\in{\tt Sh}(T,{\tt TAlg})$, and the morphism $i:\cF\to\cF^+$ of presheaves of algebras is a morphism of presheaves of {\small TA}-s. Moreover, if $\zvf:\cF\to\cF'$ is a morphism in ${\tt PSh}(T,{\tt TAlg})$, then $\zvf^+:\cF^+$ $\to\cF'^+$ is a morphism in ${\tt Sh}(T,{\tt TAlg})$. The same results hold, if we replace {\small TA}-s by {\small LCTA}-s.
\end{prop}

\begin{proof} We study the case $\cF\in{\tt PSh}(T,{\tt LCTAlg})$ (in particular $\cF\in{\tt PSh}(T,{\tt Alg})$).\medskip

Since sheafification is based on stalks $\cF_x$ ($x\in T$), i.e., on inductive limits, and since the inductive limit of a directed system of sets endowed with a same algebraic structure has as underlying set the inductive limit of the directed system of underlying sets, it is natural that the same result holds for sheafification functors. Let $+$ be the sheafification functor  $$+:{\tt PSh}(T,{\tt Alg})\ni\cF\mapsto\cF^+\in{\tt Sh}(T,{\tt Alg}):\op{For}\;.$$ Recall that $\cF^+(U)$ ($U\in{\tt Open}(T)$) is defined as the subset of $(\zP\cF)(U):=\prod_{x\in U}\cF_x$, which is made of those elements $\zs=(\zs_x)_{x\in U}=([s]_x)_{x\in U}$, for whom the section $s$ is constant in a neighborhood of any point of $U$. The algebra operations on $\cF^+(U)$ are naturally induced by those of the stalks. Restrictions are the obvious algebra morphisms. The morphism $i:\cF\to\cF^+$ of presheaves of algebras is defined by \be\label{ipresh}i_U:\cF(U)\ni s\mapsto ([s]_x)_{x\in U}\in\cF^+(U)\;,\ee and $i_U$ is injective, if $\cF$ is separated. If $\cG\in{\tt Sh}(T,{\tt Alg})$, the morphism $i:\cG\to\cG^+$ is an isomorphism of sheaves of algebras. Further, if $j:\cF\to\cG$ is a morphism of presheaves of algebras, the unique morphism $\bar \jmath:\cF^+\to\cG$ of sheaves of algebras, such that $\bar \jmath\circ i=j$, is given by $\bar\jmath_U=\prod_{x\in U} j_x$, where $j_x:\cF_x\to\cG_x$ is the morphism of algebras induced by $j\,$: \be\label{jsh}\bar \jmath_U:\cF^+(U)\ni ([s]_x)_{x\in U}\mapsto (j_x[s]_x)_{x\in U}\in\cG^+(U)\simeq\cG(U)\;.\ee Similarly, if $\cF,\cF'\in{\tt PSh}(T,{\tt Alg})$ and $\zvf:\cF\to\cF'$ is a presheaf morphism, the components of the sheaf morphism $\zvf^+:\cF^+\to\cF'^+$ are $$\zvf^+_U:\cF^+(U)\ni ([s]_x)_{x\in U}\mapsto(\zvf_x[s]_x)_{x\in U}\in\cF'^+(U)\;.$$

The stalk $\cF_x$ ($x\in T$) is the inductive limit algebra $\cF_x=\varinjlim_{U\ni x}\cF(U)$ of the directed system of algebras $(\cF(U),\zr^U_V)$, where $U$ is an open neighborhood of $x$ and $\zr^U_V$ the restriction from $U$ to $V\subset U$. This system is actually a directed system of {\small LCTA}-s, in the sense that the $\cF(U)$ are {\small LCTA}-s and the $\zr^U_V$ are continuous algebra morphisms. If we endow the inductive limit algebra $\cF_x$ with the final locally convex vector space topology with respect to the canonical algebra morphisms $\zp^x_U:\cF(U)\to\cF_x$ (i.e., with the finest locally convex vector space topology, for which the $\zp^x_U$ are all continuous), the limit $\cF_x$ is a {\small LCTVS}, whose multiplication is (jointly) continuous \cite[Lemma 2.2.]{MalBook}, i.e., the stalk $\cF_x$ is a {\small LCTA}.\medskip

In the following, the algebra $\cF^+(U)\subset(\zP\cF)(U)$ carries the induced topology of the product topology. Since any product of {\small LCTVS}-s and any subspace of a {\small LCTVS} are {\small LCTVS}-s, the algebra $\cF^+(U)$ is a {\small LCTVS}. The multiplication on $\cF^+ (U)$ is continuous in view of Lemma \ref{PreRes1}, so that $\cF^+(U)$ is a {\small LCTA}.\medskip

To show that $\cF^+\in{\tt Sh}(T,{\tt LCTAlg})$, it suffices to prove that any restriction $$r^U_V:(\zP\cF)(U)\ni(\zs_x)_{x\in U}\mapsto (\zs_x)_{x\in V}\in(\zP\cF)(V)$$ is continuous, i.e., that, for all $y\in V$, the map $$p^V_y\circ r^U_V=p^U_y:(\zP\cF)(U)\ni(\zs_x)_{x\in U}\mapsto \zs_y\in\cF_y$$ is continuous -- which is a consequence of the definition of the product topology.\medskip

The next step consists in proving that the morphism $i:\cF\to\cF^+$ of presheaves of algebras defined by \eqref{ipresh}, is a morphism of presheaves of {\small LCTA}-s, i.e., in proving that, for any $y\in U$, the map $$p^U_y\circ i_U=\zp_U^y:\cF(U)\ni s\mapsto [s]_y\in\cF_y$$ is continuous. This holds true by definition of the final topology on $\cF_y$.\medskip

Let $\zvf:\cF\to\cF'$ be a morphism in ${\tt PSh}(T,{\tt LCTAlg})$. To show that $\zvf^+:\cF^+$ $\to\cF'^+$ is a morphism in ${\tt Sh}(T,{\tt LCTAlg})$, it suffices to show that \be\label{CompE}\prod_{x\in U}\zvf_x:(\zP\cF)(U)\ni ([s]_x)_{x\in U}\mapsto(\zvf_x[s]_x)_{x\in U}\in(\zP\cF')(U)\;\ee is continuous. This is the case if and only if $\zvf_x:\cF_x\to\cF'_x$, $x\in U$, is continuous. In view of Lemma \ref{PreRes2}, the algebra morphism $\zvf_x$ is continuous if and only if $$\zvf_x\circ\zp^x_V=\zp'^x_V\circ\zvf_V:\cF(V)\ni s\mapsto \zvf_x\,[s]_x=[\zvf_Vs]_x\in\cF'_x\;,$$ $V\ni x$, is continuous. This condition is obviously fulfilled. \end{proof}

\begin{cor}\label{IndLimLCTA}
When equipped with the final locally convex vector space topology with respect to the canonical algebra morphisms $\zp^x_U:\cF(U)\to\cF_x$, a stalk $\cF_x$, $x\in T$, of a presheaf $\cF\in{\tt PSh}$ $(T,{\tt TAlg})$ is the inductive limit in ${\tt TAlg}$ of the directed system $(\cF(U),\zr^U_V)$. The same statement holds in $\tt LCTAlg$.
\end{cor}

\begin{proof}
Clearly, the $\zp^x_U:\cF(U)\to\cF_x$ are morphisms in $\tt (LC)TAlg$. Let $(F,p_U)$ be made of $F\in{\tt (LC)TAlg}$ and ${\tt (LC)TAlg}$-morphisms $p_U:\cF(U)\to F$, such that $p_V\circ\zr^U_V=p_U$. Due to Lemma \ref{PreRes2}, the unique $\tt Alg$-morphism $u:\cF_x\to F$, such that $u\circ\zp^x_U=p_U$, is continuous, since the $u\circ\zp^x_U=p_U$ are all continuous. Hence, the claim.
\end{proof}

\begin{prop}\label{ComplTensCoAlgMor}
Let $\za:A\to C$ and $\zb:B\to D$ be two continuous algebra morphisms between nuclear Fr\'echet algebras. Then $\za\widehat\0\zb:A\widehat\0 B\to C\widehat\0 D$ is a continuous algebra morphism between nuclear Fr\'echet algebras.
\end{prop}

\begin{proof}
Since $\za$ and $\zb$ are continuous linear maps between locally convex spaces, the map $\za\0\zb:A\0 B\to C\0 D$ is continuous linear, and its continuous extension $\za\widehat\0\zb:A\widehat\0 B\to C\widehat\0 D$ is continuous linear as well.\medskip

The source and the target of $\za\widehat\0\zb$ are nuclear Fr\'echet algebras \cite[Lemma 1.2.13]{Emerton}. It remains to show that $\za\widehat\0\zb$ respects their (continuous) multiplications. Let ${\frak a}=\sum_{i=0}^\infty a_i\0 b_i$ and ${\frak b}=\sum_{j=0}^\infty a'_j\0 b'_j$ be two elements of $A\widehat\0 B$. Using continuity, we get $$(\za\widehat\0\zb)({\frak a}\cdot{\frak b})=\lim_n\lim_m\sum_{i=0}^n\sum_{j=0}^m(\za(a_i)\za(a'_j))\0(\zb(b_i)\zb(b'_j))\;.$$ It is straightforwardly seen that $(\za\widehat\0 \zb){\frak a}\cdot(\za\widehat\0 \zb){\frak b}$ is given by the same limit of sums of tensor products. \end{proof}

\begin{prop}\label{ExtMult}
The multiplication $m:A\times A\to A$ of a nuclear Fr\'echet algebra $A$ extends to a continuous algebra morphism $\widehat m:A\widehat\0 A\to A$ between nuclear Fr\'echet algebras.
\end{prop}

\begin{proof}
We equip $A\0 A$ with the projective tensor topology. The continuous bilinear maps from $A\times A$ to $A$ correspond exactly to the continuous linear maps from $A\0 A$ to $A$. Hence, the multiplication $m$ can be viewed as a continuous linear map $m:A\0 A\to A$. The latter extends to a continuous linear map $\widehat m:A\widehat\0 A\to A$. This extension respects the (continuous) multiplications. The proof is similar to the one of Proposition \ref{ComplTensCoAlgMor}.\end{proof}

\subsection{Direct image and sheafification}\label{SheafiDirIm}

\begin{prop}\label{DirImShfifi}
Let $\cF\in{\tt PSh}(T,{\tt LCTAlg})$ and let $f\in C^0(T,T')$. There is a morphism $$\iota:(f_*\cF)^+\to f_*\,\cF^+$$ in ${\tt Sh}(T',{\tt LCTAlg})$.
\end{prop}

\begin{proof}
The assignment $$f_*\cF:{\tt Open}(T')\ni V\mapsto \cF(f^{-1}(V))\in {\tt LCTAlg}$$ together with the restrictions $$\zr^{f^{-1}(V)}_{f^{-1}(V')}:(f_*\cF)(V)\to (f_*\cF)(V'),$$ where $\zr$ are the restrictions of $\cF$ in $\tt LCTAlg$, is a presheaf $f_*\cF\in{\tt PSh}(T',{\tt LCTAlg})$. In view of Proposition \ref{LCTAShfifi}, we get $(f_*\cF)^+\in{\tt Sh}(T',{\tt LCTAlg})$. Similarly, we have $f_*\,\cF^+\in{\tt Sh}(T',{\tt LCTAlg})$.\medskip

Further, for $x\in T$ and $U\supset U'\ni x$, the $\tt LCTAlg$-morphisms $\zp^x_U:\cF(U)\to\cF_x$ and $\zp^x_{U'}:$ $\cF(U')\to\cF_x$ satisfy $\zp^x_{U'}\circ \zr^U_{U'}=\zp^x_U$. As, due to Corollary $\ref{IndLimLCTA}$, the stalk $(f_*\cF)_{f(x)}\in\tt LCTAlg$ is the inductive limit in $\tt LCTAlg$ of the directed system $$(\cF(f^{-1}(V)),\zr^{f^{-1}(V)}_{f^{-1}(V')})\quad(V\supset V'\ni f(x))\;,$$ there exists a unique $\tt LCTAlg$-morphism \be\label{StalksPush}u_x:(f_*\cF)_{f(x)}\to \cF_x\;,\ee such that $u_x\circ\zp^{f(x)}_V=\zp^x_{f^{-1}(V)}$, i.e., such that $u_x\,[s]_{f(x)}=[s]_x$, for all $s\in\cF(f^{-1}(V))$ and all $V\ni f(x)$.\medskip

For any $V\in{\tt Open}(T')$, the map \be\label{CompF}\iota_V:(\zP(f_*\cF))(V)\ni ([s]_y)_{y\in V} \mapsto (u_x\,[s]_{f(x)})_{x\in f^{-1}(V)}=([s]_x)_{x\in f^{-1}(V)}\in (\zP\cF)(f^{-1}(V))\ee is a continuous algebra morphism. Indeed, the algebraic operations are defined component-wise. For instance, the image of a product $([s]_y\cdot[s']_y)_{y\in V}$ is the product $(u_x[s]_{f(x)}\cdot u_x$ $[s']_{f(x)})_{x\in f^{-1}(V)}$ of the images. Moreover, the map $\iota_V$ is continuous if and only if the maps $$p^{f^{-1}(V)}_x\circ \iota_V:(\zP(f_*\cF))(V)\ni ([s]_y)_{y\in V} \mapsto u_x\,[s]_{f(x)}\in \cF_x\quad(x\in f^{-1}(V))\;$$ are. The preimage by $p^{f^{-1}(V)}_x\circ \iota_V$ of any open $\zw\subset \cF_x$ is the product over $y\in V$, whose factors indexed by $y\neq f(x)$ are $(f_*\cF)_y$ and whose factor indexed by $f(x)$ is the open subset $u_x^{-1}(\zw)$ of $(f_*\cF)_{f(x)}$. The preimage by $p^{f^{-1}(V)}_x\circ \iota_V$ is thus open in $(\zP(f_*\cF))(V)$.\medskip

The restriction of $\iota_V$ to $(f_*\cF)^+(V)$, still denoted by $\iota_V$, arrives in $(f_*\,\cF^+)(V)$. Assume that $([s]_y)_{y\in V}$ is implemented in a neighborhood $V_{y_0}$ of an arbitrary point $y_0\in V$ by a same section $t\in(f_*\cF)(V_{y_0})$, let $x_0\in f^{-1}(V)$, and set $U_{x_0}=f^{-1}(V_{f(x_0)})$. For any $x\in U_{x_0}$, we have $$u_x\,[s]_{f(x)}=u_x\,[t]_{f(x)}=[t]_x\;,$$ with $t\in\cF(U_{x_0})$.\medskip

The restriction $\iota_V:(f_*\cF)^+(V)\to (f_*\,\cF^+)(V)$ is obviously a morphism in $\tt LCTAlg$. Since these $\iota_V$ commute with restrictions, they define the sheaf morphism $\iota$ announced in Proposition \ref{DirImShfifi}. \end{proof}

\end{document}